%% file: main.tex
\documentclass[10pt,journal,compsoc]{IEEEtran}
\ifCLASSOPTIONcompsoc
  \usepackage[nocompress]{cite}
\else
  \usepackage{cite}
\fi
\ifCLASSINFOpdf
  \usepackage[pdftex]{graphicx}
   \graphicspath{{./figures/}}
\else
   \usepackage[dvips]{graphicx}
   \graphicspath{{./figures/}}
\fi
\usepackage{amsmath}
\ifCLASSOPTIONcompsoc
 \usepackage[caption=false,font=footnotesize,labelfont=sf,textfont=sf]{subfig}
\else
 \usepackage[caption=false,font=footnotesize]{subfig}
\fi
\usepackage{xcolor}
\usepackage{amsthm,amssymb}
\usepackage[ruled,vlined,linesnumbered]{algorithm2e}
\usepackage{float}
\usepackage{dirtytalk}
\usepackage{url}
\usepackage{enumitem}
\newlist{steps}{enumerate}{1}
\setlist[steps, 1]{label = Step \arabic*:}

\newcommand\T{\rule{0pt}{2.6ex}}       
\newcommand\B{\rule[-1.2ex]{0pt}{0pt}}

\theoremstyle{plain}
\newtheorem{theorem}{Theorem}[section] 

\theoremstyle{definition}
\newtheorem{definition}{Definition}[section]

\newtheorem{case}{Case}

\makeatletter
\newcommand{\removelatexerror}{\let\@latex@error\@gobble}
\makeatother

\begin{document}
%
\title{Measuring Node Contribution to Community Structure with Modularity Vitality}
%
%
%
%


\author{Thomas Magelinski,
        Mihovil Bartulovic,~\IEEEmembership{Graduate Student Member,~IEEE},
        \\
        and~Kathleen M. Carley,~\IEEEmembership{Fellow,~IEEE}
\IEEEcompsocitemizethanks{\IEEEcompsocthanksitem All authors are with Carnegie Mellon University, 5000 Forbes Avenue, Pittsburgh, PA, 15217 \protect\\
Corresponding E-mail: tmagelin@andrew.cmu.edu}
\thanks{© 20xx IEEE. Personal use of this material is permitted. Permission
from IEEE must be obtained for all other uses, in any current or future
media, including reprinting/republishing this material for advertising or
promotional purposes, creating new collective works, for resale or
redistribution to servers or lists, or reuse of any copyrighted
component of this work in other works.}}

%
%

\markboth{This paper has been accepted by IEEE Transactions on Network Science and Engineering for publication}%
{Magelinski, Bartulovic, and Carley}
%



\IEEEtitleabstractindextext{%
\begin{abstract}
Community-aware centrality is an emerging research area in network science concerned with the importance of nodes in relation to community structure.
Measures are a function of a network’s structure and a given partition.
Previous approaches extend classical centrality measures to account for community structure with little connection to community detection theory.
In contrast, we propose cluster-quality vitality measures, i.e., modularity vitality, a community-aware measure which is well-grounded in both centrality and community detection theory.
Modularity vitality quantifies positive and negative contributions to community structure, which indicate a node’s role as a community bridge or hub.
We derive a computationally efficient method of calculating modularity vitality for all nodes in $O(M+NC)$ time, where $C$ is the number of communities.
We systematically fragment networks by removing central nodes, and find that modularity vitality consistently outperforms existing community-aware centrality measures.
Modularity vitality is over 8 times more effective than the next-best method on a million-node infrastructure network.
This result does not generalize to social media communication networks, which exhibit extreme robustness to all community-aware centrality attacks.
This robustness suggests that user-based interventions to mitigate misinformation diffusion will be ineffective.
Finally, we demonstrate that modularity vitality provides a new approach to community-deception.
\end{abstract}

}

\maketitle

\IEEEdisplaynontitleabstractindextext

%
\IEEEpeerreviewmaketitle

\IEEEraisesectionheading{\section{Introduction}\label{sec:introduction}}

%
%
%
%

\IEEEPARstart{M}{odular} structure is a key phenomenon in the study of real-world networks.
Networks from a wide array of disciplines exhibit modular structure, meaning that nodes tend to be found in well-connected groups\cite{girvan2002community}.
Discovery of these clusters have been repeatedly shown to be meaningful within their context though empirical studies \cite{hartwell1999molecular,palla2005uncovering,krause2003compartments}.
Further, a ``No Free Lunch" theorem has been proved for community detection, stating that no algorithm can uniquely solve community detection, and implying that multiple valid community definitions can exist for a single network \cite{peel2017ground}.

Another fundamental question in Network Science is that of centrality.
Put simply, how important is each node in a network?
Many centrality measures have been defined over the years, each measuring ``importance" in a different way.
Classically, centrality measures are defined to be a graph invariant. 
However, network communities have been shown to be pervasive in nature, and it has been shown that networks can have multiple meaningful definitions of communities.
So, it is natural to ask the question: how important is each node in a network \textit{given some definition of groups}?
When group structure is considered, the relative importance of nodes may change. 
For example, a fairly average node in classical terms may be a hub within a small community, boosting its importance within this context.
In this work we refer to centrality measures accounting for community structure as ``community-aware centrality measures."
The question of community-aware centrality lies at the intersection of the fundamental areas of centrality and community structure.
As such, applications to community-aware centrality are far-ranging.
Here, we show applications to immunization strategies for infectious disease, robustness testing for large infrastructure networks, and privacy-based data filtering strategies.

Most of the existing community-aware centrality measures extend classic centralities by considering within-community links and between-community links separately, before applying a weighted sum to get a single score \cite{ghalmane_centrality_2019,ghalmane2019immunization,gupta2016centrality}.
This approach acknowledges the difference between links which fall within communities and those which fall between them, but ultimately gives no insight into what role a node is playing; hub-nodes and bridge-nodes can receive similarly high values without a way to distinguish them.
Further, the weighting schemes to date have been hand-crafted, rather than derived from existing community theory, making them somewhat subjective.
Cherifi et al. have acknowledged that there is room for improvement on this front \cite{cherifi2019community}.

When discussing the modularity matrix, Newman introduced ``community-centrality," which measures a node's potential to contribute to group structure \cite{newman2006finding}.
Since the measure was of \textit{potential} contribution, community-centrality is a classical centrality-measure, independent of any defined partition.
To obtain a community-aware centrality from a similar line of reasoning, we propose to measure a node's actual contribution to the group structure encoded in a specific partition.
By doing so, we obtain a community-aware centrality grounded in community detection theory and free from hand-crafted weighting schemes. 

For the measure of \textit{actual} node contributions, we turn to vitalities \cite{koschutzki2005centrality}.
In their work, Kosch{\"u}tzki et al. define vitality as the difference between the value of an arbitrary real function, $f$, applied to the graph $G$ and the same function's value when applied to the graph $G$ with the vertex of interest removed.
By doing this, a single node's contribution can be measured and the observed value can be positive or negative.
This is closely related to the key-player problem, which roughly asks to what extent a network is relying on a node's presence to remain cohesive \cite{borgatti2006identifying}. 

If the graph index is chosen to be a cluster quality metric, the vitality, then, measures a node's contribution towards group structure.
There are many such cluster quality metrics in the literature to choose from \cite{leskovec2010empirical}.
Vitalities have previously only been applied to classical centrality measures, thus, they have been defined as functions that only take graphs as arguments.
Since we are interested in community-aware centralities and vitalities, we will define vitality as a function that takes a graph and its partition as arguments.

Nodes can contribute positively or negatively to community structure.
This difference is encoded in the vitality's sign, allowing us to to distinguish nodes based on their role.
Nodes which have negative cluster quality vitality are detrimental to group structure, meaning that they are connecting groups, making them a community bridge.
Similarly, positive cluster quality vitality nodes are community hubs. 

The focus of this work is on a specific cluster quality vitality - modularity vitality.
Newman's modularity is used as the objective function for many popular community detection algorithms, making it a natural choice to measure cluster quality \cite{newman2006modularity,clauset2004finding,blondel2008fast,brandes2007modularity,traag2019louvain}.
Thus, this measure has stronger grounding in community theory than those prior, with no need for a hand-crafted weighting function.
We show that manipulation of the original modularity function leads to a scalable method of calculating modularity vitality, where the calculation for all nodes scales as $O(M + NC)$ time, where $M$ is the number of links, $N$ the number of nodes, and $C$ the number of communities.

Modularity vitality was tested on generated modular networks and on two real-world networks: the Pennsylvania Road Network, and a large Twitter network collected from the discussion of the Canadian Election of 2019.
In our experiments, modularity vitality out-performs existing community-aware centralities showing potential applications for immunization strategies, control of diffusion over networks, and for robustness testing.

While other studies have demonstrated the fragility of infrastructure networks, in our first case study, we show that the road network is over 8 times more fragile than could be seen with existing community-aware centrality measures \cite{da2015fast}. 
By targeting only 1.6\% of nodes with lowest modularity vitality, the PA road network's largest component can be reduced to less than 1\% of its original size, effectively destroying the network.

In the second case-study, the social media communication network was extremely robust, as demonstrated through the ineffectiveness of all community-aware centrality attacks on the network.
The robustness of Twitter networks has serious implications for Social Cybersecurity \cite{carley2018social,NAP25335}.
One of the core areas in this emerging discipline is developing counter-measures for the mitigation of fake or misleading news on social media.
The problem of ``Fake News" has gotten more attention recently, though many basic questions in the space are left open \cite{lazer2018science}.
It is often suggested that network metrics can be used to identify points for stopping the spread of misinformation \cite{shu2019studying}.
However, our results suggest that this is not the case.
The robustness of Twitter networks suggest that even well-targeted interventions at the user level are unable to hamper the ability of information to spread.
This result is aligned with the observed phenomena that misinformation continually resurfaces on social media \cite{shin2018diffusion}.

Lastly, we show that modularity vitality can be used to perform greedy attacks to decrease modularity. 
This gives an alternate approach to the community-deception problem, which seeks to obscure communities from detection algorithms by altering network links in order to preserve privacy. 
Modularity vitality was used to perform community-deception on a large twitter network.
The method decreased modularity by 41\%, however this decrease comes at the cost of 2\% of nodes and 45\% of edges.
While a removal of 2\% of nodes leads to a sizable decrease in modularity, this process has diminishing returns.
This suggests that a scalable and effective strategy for community deception is to obscure which popular accounts a user follows.
This differs from the typical strategy, which rewires edges instead of deleting them.

\section{Prior Work}
\subsection{Preliminaries}
Before describing the prior work, we begin with the notations and definitions that we will rely on for the remainder of the work.
\begin{definition}[Graphs] A graph is a pair $G = (V, E)$ where $V$ is a set of nodes or vertices, and $E$ of is a set of edges or links. Let us denote  $N = |V|$ as the total number of nodes and $ M = |E|$ as total the number of edges. Let $v_i \in V$ denote a node $i$ and $e_{i,j} = (v_i, v_j, w_{i,j}) \in E$ denote an edge between nodes $i$ and $j$ with weight $w_{i,j} > 0$. Finally, the adjacency matrix $A$ is is an N x N matrix with $A_{i,j} = w_{i,j}$ if $e_{i,j} \in E$ and $A_{i,j} = 0$ otherwise. For this work we only consider undirected graphs, that is $A_{i,j} = A_{j, i}.$ 
\end{definition}
\begin{definition}[Partitions] A partition of graph $G$ is $\mathfrak{C} = \{\gamma_1, \gamma_2, ..., \gamma_C\}$ where $\gamma_i$ is the set of nodes within community $i$ s.t. $\gamma_i \cap \gamma_j = \varnothing, i\neq j, \forall i,j \in \{1,\dots,C\},$ and $\gamma_1 \cup \gamma_2 \cup ...\cup \gamma_C =V$. 
We denote $C = |\mathfrak{C}|$ as the total number of communities. 
For convenience, we define a community vector, $\mathbf{c} = \left[ c_1,c_2,...,c_{N}\right]$, where $c_i$ indicates the community of node $i$. 
\end{definition}
\begin{definition}[Total and Community Degrees]
The total degree of a node $v_i$ is equal to the sum of its edges. Let us denote this by $k_i = \sum_j A_{i,j}.$ Next, define the community-degree of node $v_i$ as the sum of edges towards nodes belonging to community $c$. We denote this as 
\begin{equation*}
    k_i^c=\sum_{j=1}^N A_{i,j}\delta(c_j, c)
\end{equation*}
where the $\delta(a,b)$ is an indicator function s.t.  $\delta(a,b) = 1 $ if $a=b$, 0 otherwise. 
\end{definition}
\begin{definition}[Internal and External Degrees]
The internal degree of node $v_i$ is the sum of edges connected to $v_i$ within its community. That is $k_i^{\text{internal}}=k_i^{c_i}.$ The external degree of node $v_i$ is the sum of edges connected to $v_i$ and communities not equal to that of $v_i$. That is
\begin{equation*}
    k_i^{\text{external}} = \sum_{j=1}^N A_{i,j}(1-\delta(c_i, c_j)) = k_i - k_i^\text{internal}.
\end{equation*}
The number of internal links in the graph $G$ is given by $M^{\text{internal}} = \frac{1}{2}\sum_{i,j}A_{i,j}\delta(c_i,c_j).$
\end{definition}

\begin{definition}[Group-Fraction]
Let $G$ be a graph and $\mathfrak{C}$ be a partition of the graph $G$. The group-fraction of community $c$ is given by 
\begin{equation*}
    \mu_c= \sum_{v_i\in \gamma_c} \frac{k_i^{\text{internal}}}{k_i} = \sum_{v_i\in \gamma_c} \frac{k_i^{c}}{k_i}.
\end{equation*}
Note that this is not equal to the fraction of edges within a community.
\end{definition}

\subsection{Modularity and Grouping}
There is no ``best" way to evaluate cluster quality and as such, many cluster quality metrics have been defined \cite{leskovec2010empirical}.
While vitality measures on any of these cluster quality functions could be an interesting and unique contribution, we focus our work on modularity.
We have chosen modularity for several reasons.
First, some of the earliest discussions of community-aware centrality are given by Newman when exploring modularity \cite{newman2006finding}.
Next, many of the popular community detection algorithms attempt to maximize modularity.
Thus, studying the vitality of the quantity used to obtain the communities in the first place keeps measures consistent.
Lastly, we will show that modularity vitality in particular has a non-trivial vitality function which can be calculated efficiently.

The most common definition of modularity is that given by Newman, which is the fraction of the edges that fall within the given groups minus the expected fraction if edges were distributed at random \cite{newman2006modularity}.
The definition of Newman modularity is as follows.
\begin{definition} Given graph $G$ and partition $\mathfrak{C}$, let us define modularity as the fraction of the edges that fall within the given groups minus the expected fraction if edges were distributed at random \cite{newman2006modularity}. 
We can write modularity $Q$ of the graph $G$ as:
\begin{equation}
    Q(G, \mathfrak{C}) = \frac{1}{2M} \sum_{i, j} \left( A_{i, j} - \frac{k_i k_j}{2M} \right) \delta(c_i, c_j),
    \label{eqn:modularity}
\end{equation}
\end{definition}


Modularity in this form has been studied extensively, and the most commonly used community detection algorithms seek to maximize this quantity \cite{brandes2007modularity}.
Because it is an NP-hard problem, many different methods have been proposed to varying degrees of success \cite{blondel2008fast,clauset2004finding,traag2019louvain}.
The Louvain method has prevailed for years, and has repeatedly been shown to give meaningful communities in empirical studies \cite{blondel2008fast}.

However, recently, Traag, Waltman, van Eck have shown a flaw in the Louvain method \cite{traag2019louvain}.
Because of its update step, Louvain does not guarantee that its communities are internally connected.
It was shown that, in fact, many communities are often not connected when using the method on real-world datasets.
To fix this, Traag, Waltman, and van Eck have proposed Leiden grouping, which is slightly faster than Louvain, guarantees well-connected communities, and often achieves higher modularity. 
As such, we proceed using Leiden grouping.

\subsection{Network Centrality Measures}
Newman began the discussion of centrality based on community structure when studying the modularity matrix \cite{newman2006finding}.
He defined ``community-centrality" based on the eigenvectors of the modularity matrix.
Despite its name, this is a classical centrality measure.
Instead of measuring the actual contribution of a node, community-centrality measures a node's \textit{potential} to impact modularity.
The derivation from the modularity matrix give community-centrality a strong theoretical link to communities, but has some drawbacks.

First, potential impact can be very different from actual impact. 
A related second point is that methods which only use graph structure are unable to adapt to different graph partitions, which is significant given that  networks can have multiple meaningful definitions of communities.
Lastly, there are some practical issues.
The modularity matrix is dense, making it memory inefficient. 
Additionally, approximations are typically needed for computation on large graphs. 


Masuda takes an eigenvalue approach to achieve a community-aware centrality, though not one derived from modularity \cite{masuda2009immunization}.
Instead, he builds off of the idea of dynamical importance as defined by Restrepo et al \cite{restrepo2008weighted}.
The largest eigenvalue of a graph's adjacency matrix is related to the ease of diffusion over the graph.
Based on this fact, dynamical importance orders nodes based on the change in largest eigenvalue from the node's removal.
To leverage group structure, Masuda applied this strategy to the group-to-group network, calling it the ``mod-strategy."
This method is computationally efficient since only the largest eigenvalue is needed, and because it is calculated on the group network, which is far smaller than the actual network.
Formally, nodes were ordered based on the following equation:
\begin{align}
    \text{Mas}_i &= (2\tilde{u}_{c_i} - x ) \sum_{c \neq c_i} \tilde{u}_c k_i^c 
    \\
    x &= \frac{1}{\tilde{\lambda}} \sum_{c \neq c_i}\tilde{u}_c k_i^c,
\end{align}
where $\tilde{\lambda}$ is the group network's largest eigenvalue, and $\mathbf{\tilde{u}}$ is its corresponding eigenvector.
Intuitively, the value of the eigenvector corresponds to the importance of that group.
Thus, Masuda's method gives importance to nodes based on the group it belongs to, and its connectivity to other important groups.
The more connections to important groups, the higher the score, meaning that nodes bridging communities will be ranked highly.

More recently, degree-based measures have taken favor, due to their interpretable form and their scalability.
To get at the relationship within and between communities, these measures use internal degree and external degree.

One of the earlier examples is ``commn-centrality," $CC$, proposed by Gupta et al \cite{gupta2016centrality}.
This centrality is defined as follows:
\begin{align}
\begin{split}
    CC_i = &\left(1 - \frac{\mu_{c_i}}{|\gamma_{c_i}|}\right) \frac{k_i^\text{internal}}{\max_{v_j\in \gamma_{c_i}} k_j^\text{internal}} \times R_{c_i} + \\
    &\left(1 + \frac{\mu_{c_i}}{|\gamma_{c_i}|}\right) \left(\frac{k_i^\text{external}}{\max_{v_j\in \gamma_{c_i}} k_j^\text{external}} \times R_{c_i} \right)^2 
\end{split}
\end{align}
where $R_{c_i}$ is user-defined, but is commonly chosen as $R_{c_i}= \max_{v_j\in \gamma_{c_i}} k_j^\text{internal}$.
The group fraction $\mu$ is used so that internal degree takes precedence for weak groups, and out degree takes precedence for strong groups.
One issue with commn-centrality, however, arises when a community is disconnected from the rest of the graph.
In such a case, $\max_{v_j \in \gamma_c} k_j^\text{external} = 0,$ so commn-centrality is undefined.
This commonly occurs, especially during network robustness testing, so we do not consider commn-centrality in our experiments.

Afterward, Ghalmane et al. have proposed a number of alternatives which are well defined for community components \cite{ghalmane2019immunization,ghalmane_centrality_2019}.
The simplest of which is the number of neighboring communities centrality, which just counts the number of communities in a node's immediate neighborhood; we will call it $b_i$.
Expanding on this, the community hub-bridge centrality, $\text{CHB}$ was defined as:
\begin{align}
    \text{CHB}_i = |\gamma_{c_i}|k_i^\text{internal} + b_i  k_i^\text{external}
\end{align}
where, again, $b_i$ is the number of communities neighboring node $i$. \cite{ghalmane2019immunization}.
The number of neighboring communities centrality was out-performed by the more sophisticated community-hub-bridge centrality, so we omit it from our results to preserve readability.

Generalizing this approach beyond just degree, Ghalmane et al. introduce ``modular-centrality".
They note that a graph $G$ can be decomposed into $G^\text{internal}$ and $G^\text{external}$, where only the internal or external links are retained, respectively.
Then, internal centrality can be calculated as: $\Gamma^\text{internal}(G) = \Gamma(G^\text{internal})$, where $\Gamma$ is a classical centrality measure.
The same logic can be used to obtain external centrality.
It can be seen that when $\Gamma$ is selected to be the degree, we get the same internal and external degree as we have previously defined. 
Modular centrality is a two-dimensional vector encoding internal and external centrality.
Ghalmane et al. note that there are many ways that this vector can be used to obtain a single number, as is needed for ranking tasks.
One of their proposed methods is the weighted modular centrality, $\text{WMC}$, which takes a weighted sum of the components:
\begin{align}
    \text{WMC}_i = \mu_{c_i} \Gamma_i^\text{internal} + (1-\mu_{c_i}) \Gamma_i^\text{external}
\end{align}
where $\mu_{c_i}$ is, again, the group fraction for community $c_i$.
Note that this is the opposite weighting scheme as Gupta's; when communities are strong, modular-centrality places preference to internal degrees.
We also see Masuda's weighting giving preference to bridges.
To cover the full spectrum of these previous community-aware centralities, we also consider an adjusted version of modular-centrality, $\text{AMC}$, where the weighting scheme favors bridges:
\begin{align}
    \text{AMC}_i = (1-\mu_{c_i}) \Gamma_i^\text{internal} + \mu_{c_i} \Gamma_i^\text{external}
\end{align}
Note that this has also been previously defined as  as ``Weighted Community Hub-Bridge" centrality \cite{ghalmane2019immunization}.
Due to the similarity of the measures, we will continue using the name ``Adjusted Modular Centrality."
We also note that Ghalmane's work has been extended to overlapping communities, however this work only considers non-overlapping community structure \cite{ghalmaneOverlapping}.

With the exception of Masuda's work, these methods all rely on a weighting scheme of internal and external centrality.
The weightings are not derived from network-theoretic principles, but are based on observations seen in network studies.
Ideally, a centrality would be derived from established theory, and would eliminate the need for comparison of subjective weighting.
While Masuda's measure is derived from network theory, it is based on network connectivity, rather than community detection.

For this work, we will look to vitalities, which measure a node's contribution to some global property \cite{koschutzki2005centrality}.
\begin{definition}[Network Vitality]
\label{def:vitality}
For an arbitrary real function $f : \mathcal{G} \to \mathbb{R}$ defined on graph space $\mathcal{G}$ we write the associated vitality $V_f$ as:
\begin{align*}
    V_f(G, i) = f(G) - f(G - \{i\}),
\end{align*}
for any $G(V,E) \in \mathcal{G}$ and $i \in V$. Where $G - \{i\}$ denotes the graph $G$ after the removal of node $i$.
\end{definition}

To the best of our knowledge, vitality measures of cluster-quality functions have yet to be studied.
When cluster-quality functions are considered, the graph index must also be a function of the network partition, $\mathfrak{C}.$
Here, we select $f$ to be modularity, giving modularity vitality.

\begin{definition}[Community-Aware Vitality]
Extending the Definition \ref{def:vitality} we can write the community-aware vitality as:
\begin{align*}
    V_f(G, \mathfrak{C}, i) = f(G, \mathfrak{C}) - f(G - \{i\}, \mathfrak{C} - \{i\})
\end{align*}
\end{definition}

Through manipulation of the modularity equation, we show the calculation of modularity vitality for all nodes has time complexity of $O(M + NC)$, providing the scalability of measures like commn and modular, while maintaining the theoretic link to community detection.
At the same time, our modularity-derived measure is signed.
Negative values indicate nodes are detracting from group structure, and are thus acting like community bridges.
Positive valued nodes are then more hub-like.
Thus, unlike other measures, modularity vitality shows both how central a node is \textit{and} what way the node is central.

\subsection{Evaluation: SIR Models and Network Robustness}
\label{sec:evalRob}
Evaluation of centrality measures can be subjective, since different measures may be useful for different tasks.
However, many of the prior community-aware centrality measures have been evaluated from an immunology perspective \cite{ghalmane_centrality_2019,ghalmane2019immunization,gupta2016centrality,masuda2009immunization}.
In this scenario, a disease is spreading over a network.
The centrality measure in question is used to determine which nodes are given immunity.
Then, the ``best" centrality measure is that which leads to the smallest outbreak.
The fundamental assumption is that central nodes will be spreaders, so immunizing them should result in smaller outbreaks.

Typically, the most basic epidemic model is used: the SIR model \cite{kermack1927contribution}.
In this model, each node is either susceptible, infected, or recovered.
After an initial node is infected, it infects in neighbors with probability $p$.
At the same time, the infected nodes can recover with probability $r$.
Recovered nodes are no longer susceptible, and can no longer spread the disease.
The simulation is iterated on until there are no infected nodes remaining.
The number of nodes that were ever infected is called the epidemic size.
By immunizing nodes, the epidemic size can be decreased.
It is the goal, then, to pick an immunization strategy that leads to the smallest epidemic size.

Simulations of this type are closely related to the sub-field of Network Robustness \cite{callaway2000network,scott2006network}. 
Network Robustness refers to how a network responds to attacks. Understanding how networks react with missing nodes or edges has important implications in many fields, including but not limited to biology and ecology.
Attacks typically take the form of removal of edges or removal of nodes. We will focus on removal of nodes. 

One method of evaluating an attack's effectiveness is through network fragmentation \cite{cunha_fast_2015}. 
Fragmentation $\sigma$ can be defined as the size of the remaining largest component $N_\rho$ relative to the initial size of the graph, $N$, where $\rho$ is the fraction of nodes removed. 
Fragmentation can then be given as $\sigma(\rho)=\frac{N_\rho}{N}$. 
This is a useful measure because networks often rely on connectivity to function properly.
Disconnected components in biological, communication, or power-grid networks are in serious danger of failing completely.

Now, we can see that immunization strategies are effectively network attacks.
By immunizing a node, it and its links are removed from the network.
Immunizing many nodes fragments the network, slowing diffusion.
In fact, the fragmentation, $\sigma$, is the worst-case scenario for an SIR model.
Given the right parameterization, the disease in an SIR model will infect all nodes in the component the disease initialized in.
This behavior is guaranteed with $p=1$, and $r=0$, indicating full infection with no possibility of recovery.
The same effect can be achieved with other parameters depending on how the simulated interactions play out.
If the initial node is in the largest component, the worse-case scenario is that all nodes in the largest component get infected.
Thus, $\sigma$ can be used to measure the effectiveness of an immunization strategy without the need for expensive SIR simulations.

Replacing simulated network flow with network connectivity also results in a more fair comparison between network metrics.
Centrality measures often make assumptions about how flow occurs in a network, and are thus favorable when simulated flow matches those assumptions, and less favorable when they do not \cite{borgatti2005centrality}.
Thus, a fragmentation approach does not bias the results towards centrality measures which are best aligned with the assumptions of the simulation.

From a network robustness perspective, different types of attacks have been developed.
In general, a centrality measure is calculated for each of the nodes, and the node with the highest centrality is removed, or immunized. 
Early studies looked at node attacks based on degree \cite{callaway2000network}.
Later, Holme generalized this idea along with two styles of attacks: initial and recomputed \cite{holme2002attack}.
In the initial case, centralities are calculated once and the top-k nodes are removed.
In the recomputed case, centralities are recomputed each time a node is removed.
This makes the attack more expensive to compute, but more effective.

In this framework, attacks are defined by two characteristics, the centrality measure and the style.
Common choices of centrality measure are degree and betweenness.
Betweenness has been shown to be much more damaging to a network, but is far more expensive to compute \cite{da2015fast,holme2002attack}.
The shorthand for these methods are based on the acronym of the centrality and style; IB means an attack using initial calculation of betweenness centrality, while RD is recomputed degree. 

A connection between the modular structure in networks and their robustness has been illustrated by da Cunha et al. in \cite{da2015fast}.
The authors developed a more complex attack strategy which is able to fragment real-world networks far more quickly than the simple methods previously described.
They achieve this by ensuring that nodes are attacked only when they are in the largest component and when they are connecting groups. 
This strategy is called a Module-Based-Attack, MBA.

Though effective, attacks using betweenness centrality do not scale to the size of networks commonly seen on social media.
For weighted networks, a single calculation of betweenness scales as $O(NM + N^2\log N)$, making RB scale as $O(N^2 M + N^3 \log N)$ \cite{brandes2001faster}.
This makes RB intractable for medium-sized networks, which is why da Cunha et al. use IB as the base for their attack method \cite{da2015fast}.
However, even IB is intractable for very large networks.
Additionally, the computation of largest component at every step adds to the method's complexity.
The most scalable methods are those that use an ``initial" strategy with a local measure.

Based on this, we use fragmentation to evaluate our method in comparison to the following measures:
Masuda (Mas), Community-Hub-Bridge (CHB), Modular-Centrality-Degree (WMC-D), Adjusted-Modular-Centrality-Degree (AMC-D), and Degree (Deg).
Evaluation is performed in three steps.
In the first, networks are generated to measure how different community-aware centralities perform under varying attack strategies.
In this step ``initial," ``repeated," and ``module-based" attacks are performed.
Second, the Pennsylvania road network is studied.
This is a large highly modular network, which exemplifies the power of community-aware centrality measures.
Finally, a large Twitter communication network is studied from the Canadian Elections of 2019. 
Here, the robustness of social media networks is demonstrated.
In the second and third steps, only ``initial" strategies are taken due to the size of the networks.

\subsection{Community Deception}
Community Deception has recently been formalized by Fionda and Pirro \cite{fionda2017community}.
They argue community detection is a very powerful tool, and could potentially be too powerful for privacy-sensitive applications. 
In order to protect sensitive data that is easily identifiable, community structure should be obscured. 
The goal, then, is to edit a network to prevent a specific community's detection. 
The most relevant framing they provided to the present work is Modularity-Based Deception.
In this framing, the goal is to re-wire edges such that modularity of a community is minimized. 
This approach is based on the modularity equation, similarly to the present work, and is scalable.
In a similar line of work, Chen et al. propose a genetic algorithm to perform a ``Q-Attack," which edits the network to minimize the modularity of the entire network's partition, not just that of a single community \cite{chen2019ga}. Due to the combinatorial nature of genetic algorithms this approach did not scale and was only tested on nodes with approximately 100 nodes.

Waniek et al. also consider the single-community case \cite{waniek2018hiding}.
In this work, a modularity-inspired measure was used to determine how well a community is concealed.
The authors then randomly rewire a specified number of internal edges as external edges.
This approach demonstrated that social network users had the power to conceal their community from detection.
However, the method is non-deterministic, so its effectiveness varies depending on which edges were selected in each round of simulation.
The lack of distinction between the best edges to add or remove also makes it difficult for users to best select actions to conceal their community.

Lastly, Nagaraja takes a different view of the problem wherein an adversary is attempting to uncover the community structure of the entire network with a surveillance strategy \cite{nagaraja2010impact}.
The work proposes several counter-strategies to conceal communities with edge alterations.
Nagaraja concludes that these strategies work based on how they impact the network's modularity, without explicitly maximizing for impact on modularity.
The present enables this to be explicitly maximized.

Here, we show that Modularity Vitality can be used to efficiently perform community deception on the entire network rather than a specific community.
Rather than rewiring edges, we remove all edges attached to nodes with the highest modularity vitality.
This has the benefit of keeping maintaining network accuracy for links that are present, but ultimately does change the degree distribution. 
In a social media setting, this amounts to hiding which popular accounts a user follows, rather than re-wiring individual following relationships.
We demonstrate the power of this approach by performing community deception on a social media communication network with 7.5 million nodes, and 130 million edges.

\section{Calculating Modularity Vitality}
Newman's community centrality measured a node's \textit{potential} to contribute to modularity.
To calculate the \textit{actual} contribution, we can calculate the modularity vitality: the difference between the modularity of the original partition, and the modularity of the partition after the removal of a specified node.
Given that community-aware centralities are commonly evaluated using the effect of node removals, modularity vitality seems to be a natural approach.
Note that once a node is removed, the network could be re-grouped, and the group structure could potentially be quite different.
Once regrouping is considered, there is no closed-form solution to what the new modularity would be, since the maximization procedure would need to be re-run. 
Thus, regrouping is typically not considered, and we do not consider it here \cite{masuda2009immunization}.

Modularity vitality is defined as:
\begin{equation}
    V_Q(G, \mathfrak{C}, i) = Q(G, \mathfrak{C}) - Q(G - \{i\}, \mathfrak{C}- \{i\}).
    \label{eq:mod_vitality}
\end{equation}
A naive computation of this expression is quite expensive.
Modularity itself has time complexity $O(M)$.
Thus, naively recalculating this in order to calculate the modularity vitality for all nodes has complexity $O(MN)$.
However, there is an efficient way of updating modularities after the removal of a node.

The modularity after the removal of node $i$ can instead be calculated using the following expression:
\begin{align}
\begin{split}
        &Q(G - \{i\}, \mathfrak{C}- \{i\}) = \\[1em]
        &\frac{M^\text{internal}-k_i^\text{internal}}{M-k_{i}} - \frac1{4 \left ( M-k_{i}\right )^2} \sum_{\gamma_c \in \mathfrak{C}} \left ( d_c - h_{i,c}\right )^2
        \label{eqn:new_modularity}
\end{split}  \\
        &h_{i,c} = k_{i}^{c}  + k_i \delta (c,c_{i}).
\end{align}

We will now derive this equation.
\label{sec:mv_proof}
\begin{theorem} 
If we remove node $i$ from the graph $G$ then the new modularity of the new graph $G-\{i\}$ can be written as:
\begin{align}
\begin{split}
        &Q(G - \{i\}, \mathfrak{C}- \{i\}) = \\[1em]
        &\frac{M^\text{internal}-k_i^\text{internal}}{M-k_{i}} - \frac1{4 \left ( M-k_{i}\right )^2} \sum_{\gamma_c \in \mathfrak{C}} \left ( d_c - h_{i,c}\right )^2
\end{split}  \\
        &h_{i,c} = k_{i}^{c}  + k_i \delta (c,c_{i}).
\end{align}
The value $h_{i, c}$ measures the number of edges a node has to that community, and if the node is a member of said community, its degree is added. 
The degree must be added because $d_c$ double-counts the number of internal links in a community.
\end{theorem}

\begin{proof}
\noindent The removal of node $i$ from graph $G$ results in a new graph denoted by $G - \{i\}$. The same applies to the community vector, which is denoted by $\mathfrak{C}- \{i\}$.

First, we re-write Modularity as given in Equation \ref{eqn:modularity}:
\begin{align*}
    Q(G, \mathfrak{C}) =& \frac1{2M}\sum_{i,j=1}^N \left (  A_{i, j} - \frac1{2M} k_i k_j \right ) \delta(c_i,c_j)
    \\
    =& \frac1{2M}\sum_{\gamma \in \mathfrak{C}} \sum_{v_i,v_j \in \gamma} \left (  A_{i,j} - \frac1{2M} k_i k_j \right ) 
    \\
    =& \frac1{2M} \underbrace{\sum_{\gamma\in\mathfrak{C}} \sum_{v_i,v_j \in \gamma} A_{i,j}}_{2M^\text{internal}} - \frac1{4M^2}\sum_{\gamma\in\mathfrak{C}} \sum_{v_i,v_j \in \gamma} k_i k_j 
    \\
    =& \frac{M^\text{internal}}{M} - \frac1{4M^2}\sum_{\gamma\in\mathfrak{C}} \sum_{v_i,v_j \in \gamma} k_ik_j 
    \\
    =& \frac{M^\text{internal}}{M} - \frac1{4M^2} \sum_{\gamma \in \mathfrak{C}} \sum_{v_i \in \gamma} k_i \sum_{v_j \in \gamma} k_j
\end{align*}

\noindent Let
\begin{equation*}
   d_c = \sum_{v_i \in \gamma_c} k_i = \sum_{v_j \in \gamma_c}k_j
\end{equation*}

\noindent Now can express modularity in terms of number of links and total degrees of nodes:

\begin{equation}
    Q(G, \mathfrak{C}) = \frac{M^\text{internal}}M - \frac1{4M^2} \sum_{\gamma_c\in\mathfrak{C}} d_c^2
    \label{eqn:simp_mod}
\end{equation}
This form is easier to derive the new modularities from.

Equation \ref{eqn:simp_mod} can then be applied to graph on graph $G - \{i\}$ to find the new modularity:
\begin{align*}
&Q(G - \{i\}, \mathfrak{C} - \{i\}) = \\[1em]
&\frac{M^\text{internal}-k_i^\text{internal}}{M-k_{i}} - \frac1{4 \left ( M-k_{i} \right )^2} \sum_{\gamma_c\in \mathfrak{C}} \tilde{d}_{i,c}^2
\end{align*}

\noindent Now to calculate $\tilde{d}_{i,c}$ we can break this down in two cases:

\setcounter{case}{0}
\begin{case} If $c \neq c_{i}$ we have :
\end{case}
\begin{equation*}
    \tilde{d}_{i,c} = \sum_{v_j \in \gamma_c}k_j - k_i^c
\end{equation*}

\begin{case}
If $c = c_{i}$ we have :

\end{case}
\begin{equation*}
    \tilde{d}_{i,c} = \sum_{v_j \in \gamma_c}k_j - k_i^c -k_i
\end{equation*}

\noindent Let :

\begin{equation*}
    h_{i,c} = k_i^c + k_i \delta (c,c_i) 
\end{equation*}

\noindent then finally we have:
\begin{equation*}
    \tilde{d}_{i,c} = d_c -h_{i,c}
\end{equation*}
\noindent Giving us the final expression for the modularity once node $i$ is removed:
\begin{align*}
    &Q(G - \{i\}, \mathfrak{C}- \{i\}) = \\[1em]
    &\frac{M^\text{internal}-k_i^\text{internal}}{M-k_{i}} - \frac1{4 \left ( M-k_{i}\right )^2} \sum_{\gamma_c \in \mathfrak{C}} \left ( d_c - h_{i,c}\right )^2 
\end{align*}

\end{proof}

By looking at Equation \ref{eqn:new_modularity}, we observe that the only new information needed to update modularity after removing a node is contained in the node's immediate neighborhood and the vector of community degrees. 
The worst-case scenario would be to calculate updated modularity for the center node of a star-graph, which has degree $M$.
When Equation \ref{eqn:new_modularity} is used, the  time complexity of calculating the new modularity becomes $O(M + C)$.
While this seems to not be an improvement, the worst-case scenario is far worse than the average case, since most node degrees are far less than $M$.
In fact, the calculation of Equation \ref{eqn:new_modularity} for \textit{all nodes} in a network has time complexity of only $O(M + NC)$.
Given that typically $C \ll N$, this is a major improvement over the naive implementation's complexity of $O(MN)$.
This improvement allows for analysis of very large graphs for which $O(MN)$ operations could be prohibitively expensive if not infeasible.

By studying modularity vitality, rather than just the simple new modularity after a node is removed, it is easy to identify which nodes are increasing modularity and which are decreasing it.
As Newman noted, ``it is entirely possible for individual vertices to simultaneously make both large positive and negative contributions to modularity" \cite{newman2006finding}.
A simplistic approach would be to add the absolute value of the two, but Equation \ref{eq:mod_vitality} balances them to see which contribution prevails for each node.
Since nodes with positive modularity vitality are contributing positively towards community structure, they can be thought of as hubs within their community.
Their removal decreases the strength of their communities, thus decreasing modularity.
Conversely, nodes negatively contributing to group structure will have negative modularity vitality.
Negative contributions to group structure are facilitated through connections between groups, so nodes with highly negative modularity vitality are community bridges.
Removing these community bridges increases modularity.
A measure which does not have the issue of large positive and negative contributions balancing out is explored in Appendix \ref{sec:community_degree}, though it does not perform as well as modularity vitality.

Like many previous measures, modularity vitality is correlated with degree.
This correlation is intuitive: nodes with many connections have the most potential to impact group structure, either positively or negatively.
It can be seen in the new modularity equation: as node degree increases, the denominator decreases, leading to increase in the magnitude of modularity vitality.
However, modularity vitality is more complex, since it takes into account which groups a node is connected to.
Nodes connected to larger groups have a bigger impact than those connected to smaller groups.
This mirrors Masuda's measure, where a node's importance is based on the importance of its group and the group(s) it is connected to.
The difference here is that modularity vitality measures a group's importance with the number of internal links, while Masuda's uses the eigenvector centrality with the group to group network.

\section{Methodology}
\subsection{Fragmentation-Based Evaluation}
As discussed in Section \ref{sec:evalRob}, evaluation based on network fragmentation is similar to the SIR evaluation used in other studies, like \cite{ghalmane_centrality_2019, ghalmane2019immunization}, however is less expensive computationally and is easier to interpret.
Module-based attacks (MBA's) were tested in such a framework, where they were shown to effectively fragment networks \cite{da2015fast}.
Again, fragmentation $\sigma$ is the size of the largest component after the attack, divided by the original largest component.
Fragmentation is measured as a function of $\rho$, the fraction of nodes removed in the attack: $\sigma(\rho) = \frac{N_\rho}{N}$.
Similarly, fragmentation can be looked at as a function of the fraction of edges removed, $\eta.$
Note that here we are only targeting nodes, not edges, but the fraction of remaining edges is still an interesting quantity to study, as we see in Section \ref{sec:election}.


An immunization or fragmentation strategy's effectiveness depends on how many nodes are removed, as seen by the notation $\sigma(\rho).$
To measure the overall effectiveness, the fragmentation function can be integrated.
The lower the integral, the more effective the strategy, so we will call this the cost function that we are trying to minimize: $C_\rho = \int_\rho \sigma(\rho) d\rho$.
For comparison, the cost with respect to edges can be of interest, though it is not directly being optimized: $C_\eta = \int_\eta \sigma(\eta)d\eta.$

Thus, we will evaluate all of the attack strategies in Section \ref{sec:strategies}, using $C$. 
We will do so in three parts: generated networks, the PA road network, and a Twitter network obtained from user to user conversations surrounding the Canadian Election of 2019. Each part highlights different aspects of the proposed method.

\subsection{Attack Strategies}
\label{sec:strategies}
Attack strategies are the rules that govern which nodes are to be immunized, or removed from the network.
Generally these strategies are independent of centrality measure, so can be paired with any measure of a researcher's choosing.
Holme outlined two strategies: initial and repeated \cite{holme2002attack}.
In the initial attack, a centrality measure is calculated for each of the nodes.
Then, the top-k nodes are selected to be attacked.
The procedure is outlined in Algorithm \ref{alg:Initial}.

Perhaps the biggest issue with the initial attack strategy is its redundancy.
After the first node is removed, the centralities of the following nodes change. However, these changes go un-detected in the initial attack model, leading to the selection of nodes that are no longer in central positions. 
This, to some extent, can happen due to random effects of node and edge removal in a network.
The extent to which random removals impact centrality values and rankings has been previously studied by Borgatti and others, who find that the accuracy of centrality measures drops off smoothly as the number of random changes to the graph increases, though this effect is dependent on the network's topology \cite{borgatti2006robustness, frantz2009robustness}.
Perhaps more importantly, there are non-random effects at play.
It is known that certain central nodes are responsible for the centrality of other nodes, and that this can be measured with exogenous centrality \cite{everett2010induced}.

The redundancy issue of the initial attack strategy is resolved in the recomputed attack strategy wherein the centralities are recomputed after each node removal.
The full algorithm is shown in  Algorithm \ref{alg:Recomputed}.
Though effective, the recompute step adds scalability issues.
For a centrality measure that takes $O(M)$ time, the attack takes $O(NM)$ time.
This means for expensive calculations like betweenness, the recompute strategy will be intractable, $O(N^3 \log N)$ for weighted networks \cite{brandes2001faster}.

A more sophisticated strategy is given by da Cuhna et al, called Module-Based-Attack (MBA) \cite{da2015fast}. 
The authors find that use of group structure leads to effective fragmentation.
Group-based structure is incorporated by only attacking nodes which bridge communities.
Further, only nodes in the current largest component are attacked.
While largest component is recomputed, the centrality measures are not.
The full procedure is given in Algorithm \ref{alg:MBA}, where $\bigoplus$ denotes append operation.
While not as complex as the recompute method, the update of the largest component and node bridges makes the method significantly more computationally expensive when compared to the simple initial attack.

\begin{figure}[ht]
  \centering
  \removelatexerror
  \begin{minipage}{.8\linewidth}
\begin{algorithm}[H]
\SetAlgoLined
\KwResult{List of removed nodes $\mathcal{L}$}
$\mathcal{L} \leftarrow \varnothing$\;
$k \leftarrow$  the number of nodes to remove\;
 $\mathcal{S} \leftarrow$ List of all nodes sorted by a centrality measure (function)\;
\While{ $|\mathcal{L}| < k$}{
 $\tau \leftarrow$ top node in S\;
   $\mathcal{L} \leftarrow \mathcal{L} \cup \tau$\;
   $\mathcal{S} \leftarrow \mathcal{S} \setminus \tau$\;
}
 \caption{Initial Attack}
   \label{alg:Initial}
\end{algorithm}
  \end{minipage}
\end{figure}

\begin{figure}[ht]
  \centering
  \removelatexerror
  \begin{minipage}{.8\linewidth}
\begin{algorithm}[H]
\SetAlgoLined
\KwResult{List of removed nodes $\mathcal{L}$}
$\mathcal{L} \leftarrow \varnothing$\;
$k \leftarrow$  the number of nodes to remove\;
$G \leftarrow$ the initial graph\;
\While{ $|\mathcal{L}| < k$}{
  $\mathcal{S} \leftarrow$ List of all nodes in $G$ sorted by a centrality measure (function)\;
  $\tau \leftarrow$ top node in S\;
  $\mathcal{L} \leftarrow \mathcal{L} \cup \tau$\;
   $G \leftarrow G \setminus \tau$\;
}
 \caption{Repeated or Recomputed Attack}
   \label{alg:Recomputed}
\end{algorithm}
  \end{minipage}
\end{figure}

\begin{figure}[ht]
  \centering
  \removelatexerror
  \begin{minipage}{.8\linewidth}
\begin{algorithm}[H]
\SetAlgoLined
\KwResult{List of removed nodes $\mathcal{L}$}
$\mathcal{L} \leftarrow \varnothing$\;
$G \leftarrow$ the initial graph\;
$B \leftarrow$  the set of nodes bridging communities in $G$\;
 $\mathcal{S} \leftarrow$ List of all nodes in $G$ sorted by a centrality measure (function)\;
$LC \leftarrow$  the set of nodes in the largest component of $G$\;
\While{ $|B \cap LC|> 0$}{
 $\tau \leftarrow$ top node in S\;
  \eIf{ $\tau \in B$ \textbf{and} $\tau$ $\in LC$}{
   $G \leftarrow$ $G \setminus \tau$\;
   $LC \leftarrow$  the set of nodes in the largest component of $G$\;
   $B \leftarrow$   the set of nodes bridging communities in $G$\;
   $\mathcal{L} \leftarrow \mathcal{L} \cup \tau$\;
   $\mathcal{S} \leftarrow \mathcal{S} \setminus \tau$\;
   
   }{
   \eIf{ $\tau \notin B$ } {
   $\mathcal{S} \leftarrow \mathcal{S} \setminus \tau$;\ 
   }{
   $\mathcal{S} \leftarrow \mathcal{S} \bigoplus \tau$  
   }
  }
}
 \caption{Module-Based Attack (MBA)}
   \label{alg:MBA}
\end{algorithm}
  \end{minipage}
\end{figure}

Thus, for small generated networks we take I, R, and MBA. 
For the PA-Road Network and Twitter networks, however, only the ``initial" attack strategy is computed, as it is the most scalable.
These are combined with degree as well as the previously discussed local community-aware centrality measures: Masuda (Mas), Community-Hub-Bridge (CHB), Modular-Centrality-Degree (WMC-D), Adjusted-Modular-Centrality-Degree (AMC-D), and Degree (Deg).
We compare these existing approaches to modularity vitality in two forms. 
First, we take the original modularity vitality (MV), attacking from negative to positive, in order to target community-bridges.
Second, we consider the absolute value of the modularity vitality (AMV), which targets nodes based on their overall contribution to group structure, positive or negative.
A third form was considered, where nodes were attacked from positive to negative modularity vitality.
This hub-first strategy did performed poorly, and is omitted from result tables to preserve readability.

\section{Network Fragmentation}

\subsection{Generated Networks}
First, we compared community-aware centralities using generated networks.
By using generated networks we can repeat tests many times.
We constructed modular networks using the cellular network model, similar to that of Masuda \cite{masuda2009immunization}.
In this model, ``cells" are random sized Erd\H{o}s-R\'enyi networks with high density, simulating clusters.
Then, the cell-to-cell network is also modeled as an Erd\H{o}s-R\'enyi network.
When two cells are linked in the group-to-group network, random nodes from each are selected and a link is drawn between them. For this study, cellular networks were created using the parameters  shown in Table \ref{table:cell_graph_prams}.
This results in an unweighted, undirected random network with community structure.


\begin{table}[!ht]
\caption{Cellular Network Parameters. $\mathcal{U}(a,b)$ denotes the uniform random distribution between numbers $a$ and $b$; $\mathcal{N}(\mu,\sigma^2)$ denotes the normal distribution with mean $\mu$ and variance $\sigma^2$.  }
\label{table:cell_graph_prams}
\centering
\resizebox{\columnwidth}{!}{
\begin{tabular}{|l|l|}
\hline
  \textbf{Variable}   &  \multicolumn{1}{c|}{\textbf{Description}} \\ \hline \hline
$N =    1000$                                           &  Number of nodes \T\B \\ \hline 
$N_c = \mathcal{U}(10, 20)$                          &  Number of cells \T\B \\ \hline
$n_i = \mathcal{N}(N /N_c, N_c / 5),$ &  Number of nodes per cell \T\B \\ \hline
$p_i =\mathcal{U}(0.1, 0.25)$                       & Density of internal cell relationships \T\B \\ \hline
$p_o =\mathcal{U}(0, 0.5)$                          & Density of the cell-to-cell network \T\B \\ \hline
\end{tabular}}
\end{table}

The eight previously discussed community-aware centrality measures were paired with the three possible attack schemes, initial, recomputed, and MBA, to give 24 attacks.
Each time a network was generated all 24 attacks were performed on the network and the corresponding cost functions $C_\rho$ and $C_\eta$ were recorded.
The average cost of the 24 attacks for 100 generated networks is given in Table \ref{tab:generated_results}.
The average modularity for these 100 networks when grouped with Leiden grouping was 0.91.
The modularity vitality attack consistently outperforms all other attacks both in terms of node cost and edge cost, suggesting that it is the best community-aware centrality measure for this type of synthetic network.
The fact that attacking nodes with negative modularity vitality is more effective than nodes that are high in modularity vitality magnitude suggests that community bridge nodes are more important than community hub nodes in cellular networks.
The success of the ``adjusted" modular-degree centrality provides further evidence of this, since it places greater importance on community bridges, while the original modular-degree focuses on hubs and does not score as well.

\begin{table*}[!t]
    \caption{Results for attacks on the generated cellular networks. The values shown are the average over 100 simulations. Methods introduced in this work are on the left of the double column. The best results are emboldened.}
    \label{tab:generated_results}
    \centering
    \tiny
    \resizebox{\textwidth}{!}{%
    \begin{tabular}{|l|c|c||c|c|c|c|c|}
         \hline
         \textbf{Method} & MV & AMV & AMC-D & Mas & CHB & WMC-D & Deg  \\ 
         \hline \hline
         Initial $C_\rho$ & \textbf{0.165} & 0.211 & 0.169 & 0.198 & 0.383 & 0.381 & 0.347
         \\
         \hline
         Initial $C_\eta$ & \textbf{0.247} & 0.308 & 0.268 & 0.293 & 0.576 & 0.599 & 0.578
         \\
         \hline
         \hline
         MBA $C_\rho$ & \textbf{0.086} & 0.087 & 0.088 & 0.090 & 0.101 & 0.103 & 0.100
         \\
         \hline
         MBA $C_\eta$ & \textbf{0.157} & 0.162 & 0.173 & 0.162 & 0.211 & 0.219 & 0.216
         \\
         \hline
         \hline
         Recomputed $C_\rho$ & \textbf{0.107} & 0.126 & 0.130 & 0.132 & 0.331 & 0.337 & 0.309
         \\
         \hline
         Recomputed $C_\eta$ & \textbf{0.188} & 0.205 & 0.221 & 0.205 & 0.608 & 0.616 & 0.586
         \\
         \hline
    \end{tabular}
    }
\end{table*}

\subsection{PA-Road Network}
One particularly well-suited application for community-aware centrality measures is the analysis of large infrastructure networks.
These networks typically have two properties: very high modularity and low maximum degree.
High modularity makes group-based approaches appropriate.
Low maximum degree often means that simple degree-based attacks will be ineffective.
Additionally, their large size make effective approaches like MBA intractable, or at least very costly.
Instead, we show that initial-attacks with community-aware centrality measures are very effective, and that our modularity-based methods are the most effective by far.

As an example, we use the Pennsylvania Road Network \cite{pa_roads}.
Roads are represented by edges, while intersections are represented by nodes.
This network has 1,088,092 nodes, and 1,541,898 edges. When grouped with Leiden grouping maximizing modularity, 499 clusters are obtained with a modularity of 0.990.
Its maximum degree is 18.
The extremely high modularity and low maximum degree make it an ideal candidate for community-aware centrality measures.

\begin{figure*}[htb]
    \centering
    \subfloat[Fragmentation by nodes removed.]{{\includegraphics[width=0.48\textwidth]{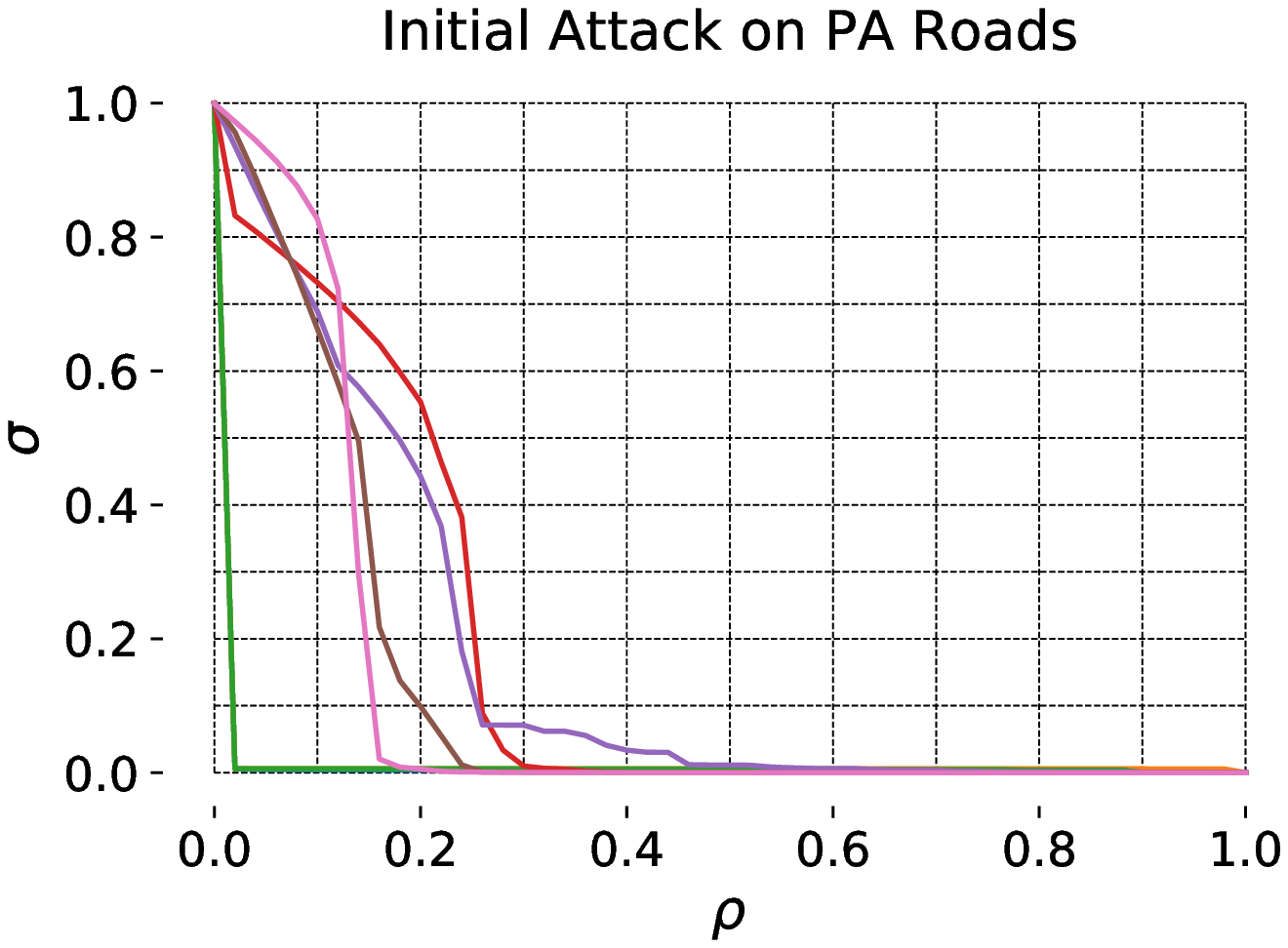} }}
    \subfloat[Fragmentation by edges removed.]{{\includegraphics[width=0.48\textwidth]{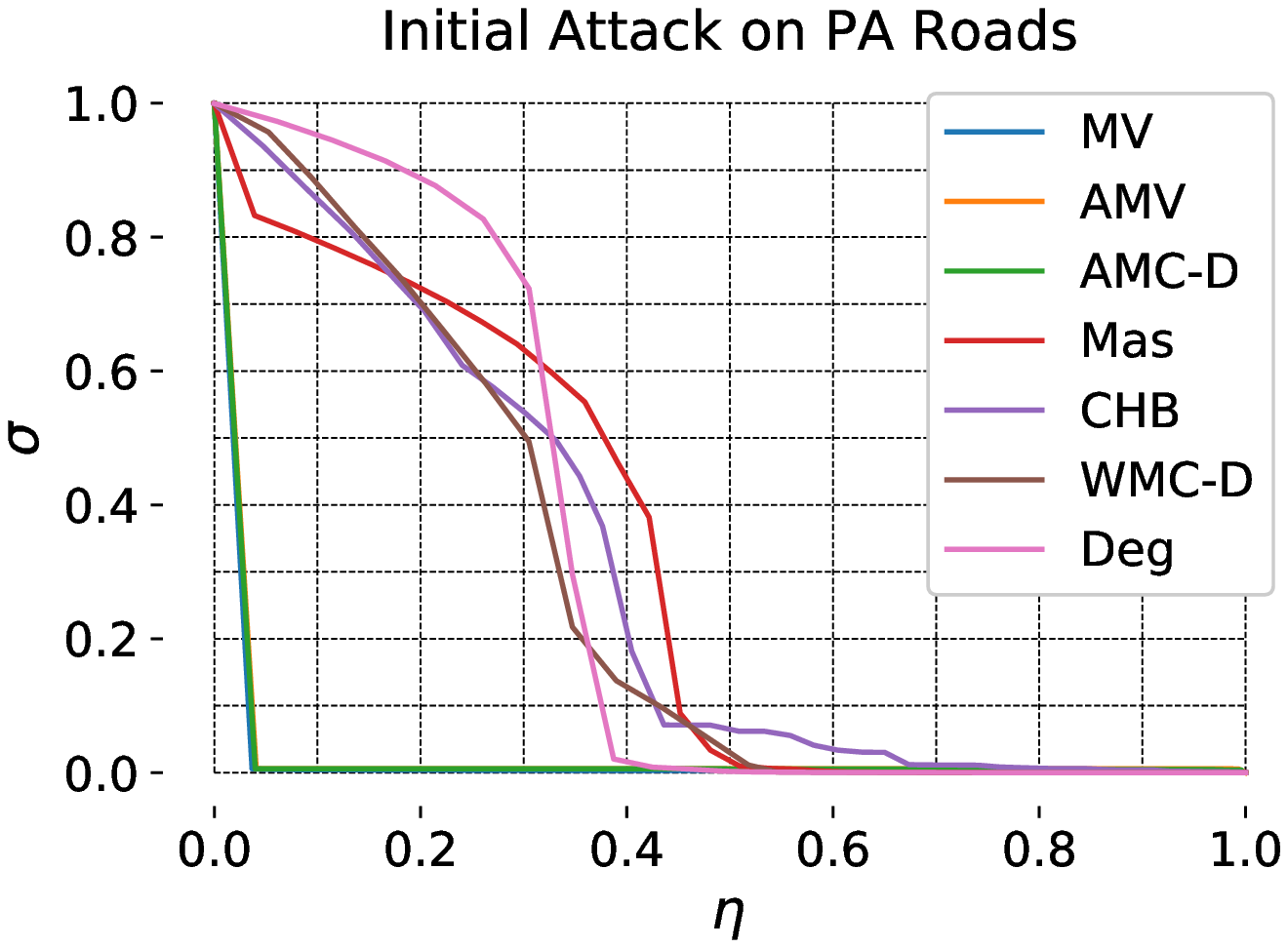} }}%
    \caption{Fragmentation of the PA-Road Network. Results for modularity-vitality, absolute modularity-vitality, and adjusted modular-centrality are extremely similar, so overlap on both figures. The legend in (b) also applies to the plot in (a), as well as both plots in Figure \ref{fig:election}.}%
    \label{fig:pa_roads}%
\end{figure*}

In Figure \ref{fig:pa_roads}, we see the fragmentation as a function of nodes and edges removed for each strategy.
Here, we see the largest component can be effectively brought to zero by removing 1.6\% of nodes with the highest modularity vitality values.
Removing only positive-valued nodes and removing nodes based on the absolute value of their modularity vitality value give very similar results.
The quantitative results, as measured by $C_\rho$ and $C_\eta$ are given in Table \ref{tab:cost_results}.

Additionally, we show the plot as a function of edges removed, for the same strategies.
The edge plot shows that while modularity vitality fragments the networks best given a number of nodes, it is also most efficient in terms of edges.

With just a degree-based attack, it would appear that the Pennsylvania road network is robust.
In fact, the community-aware centrality methods show that it is quite fragile.
Using an I-MV attack, the network can be almost completely fragmented by targeting only 1.6\% percent of nodes, bringing the largest component down to less than 1\% of its original size. 
This improves over the previous best measure, modular-degree, by a factor of over 8.5.

\subsection{Canadian Election Twitter Network}
\label{sec:election}
Another relevant application of community-aware centrality is social media networks.
Since social media has become so embedded in everyday life, scalable tools to understand it are essential.
Given the increasing polarization of online discussion, as described in concepts like filter bubbles, it is not enough to know what actors are important in general.
Instead, it is necessary to understand what actors are important within and between key online communities.
Community-aware centralities make this a measurable problem.

To study the effectiveness of our community-aware centrality measures we again use network fragmentation, due to its connection with diffusion.
Diffusion on social media is an important phenomena to understand as a way to combat misinformation, among other things.
Users who fragment the network when removed are those who have the most power to spread information. 

For this study, we use the network created from Twitter data collected during 2019 Canadian federal election.
The goal was to obtain a user-to-user communication network where users were active in political discussion.
First, we used a keyword search of Twitter's API to collect tweets related to the Canadian Election during the month of October.
From here, the unique users were recorded, giving a list of users active in political discussion.
While a user to user network could be constructed with this data, many links would be missing, since only tweets with our keywords can be used.
To construct a more complete network, Twitter's API was used to scrape the timelines of all users in our list.
This new collection was then truncated to the week of the election.
Finally, the all-communication graph was computed from this dataset, where link weights are the sum of the mentions, retweets, and quotes.
The ``Election Week" network, has 7,523,125 nodes, and 130,086,491 links.
When grouped with Leiden grouping, 557 communities were discovered, with a modularity of 0.691.

Figure \ref{fig:election} shows the fragmentation results on the election week network.
Again, the quantitative results are given in Table \ref{tab:cost_results}.
The Adjusted-Modular-Degree measure and the classical degree measure effectively tie for node-based efficiency.

The structure and properties between the PA Roads network and the Election Week network are very different.
This difference is reflected in Figure \ref{fig:election}.
Perhaps most striking is how poorly the modularity vitality method performs in terms of $\rho$.
While other methods fragment the network removing 10-30\% of nodes, the positive modularity vitality method does not fragment the network until nearly all nodes are removed.

\begin{figure*}[!ht]
    \centering
    \subfloat[Fragmentation by nodes removed.]{{\includegraphics[width=0.48\textwidth]{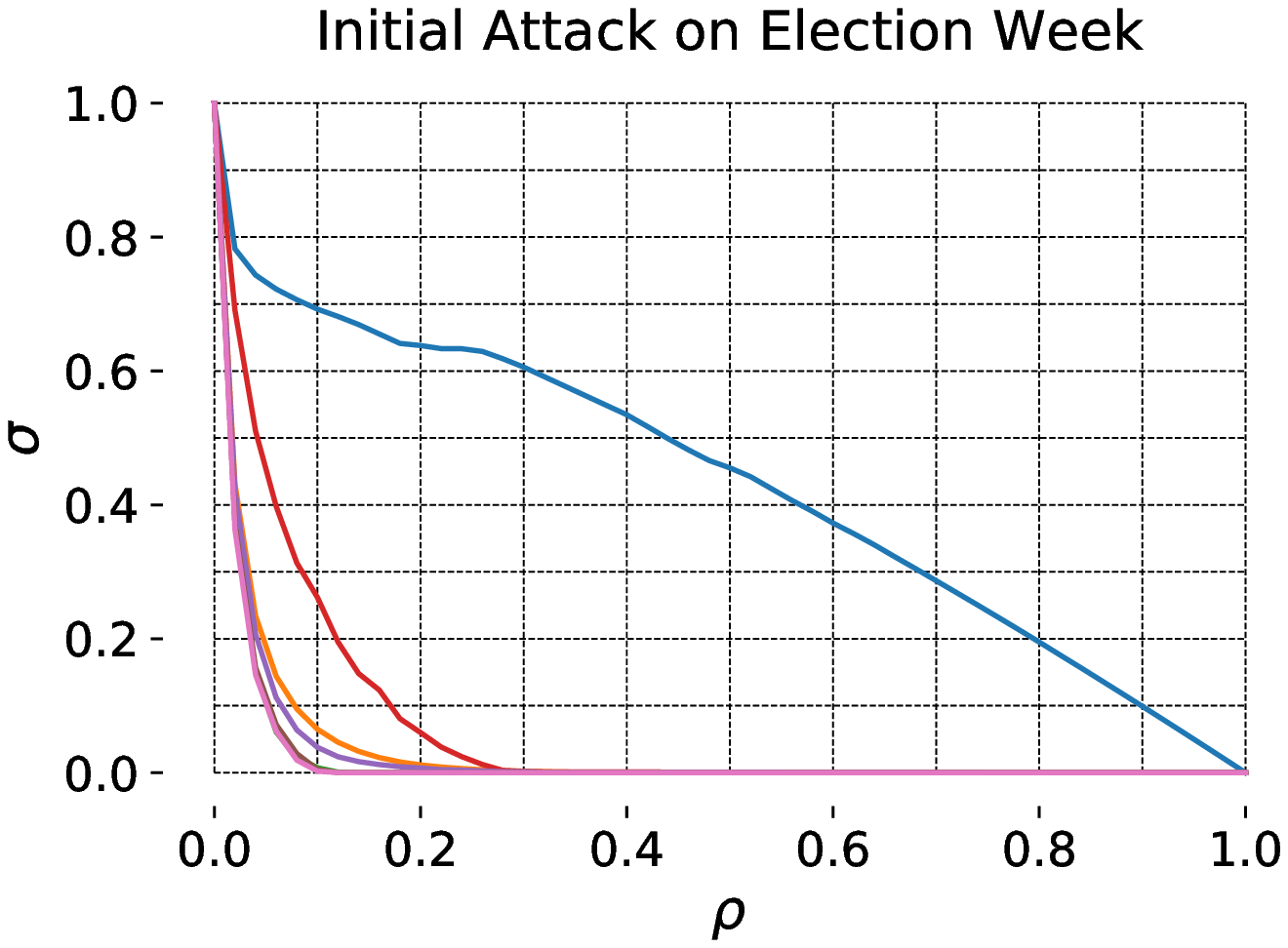} }}
    \subfloat[Fragmentation by edges removed.]{{\includegraphics[width=0.48\textwidth]{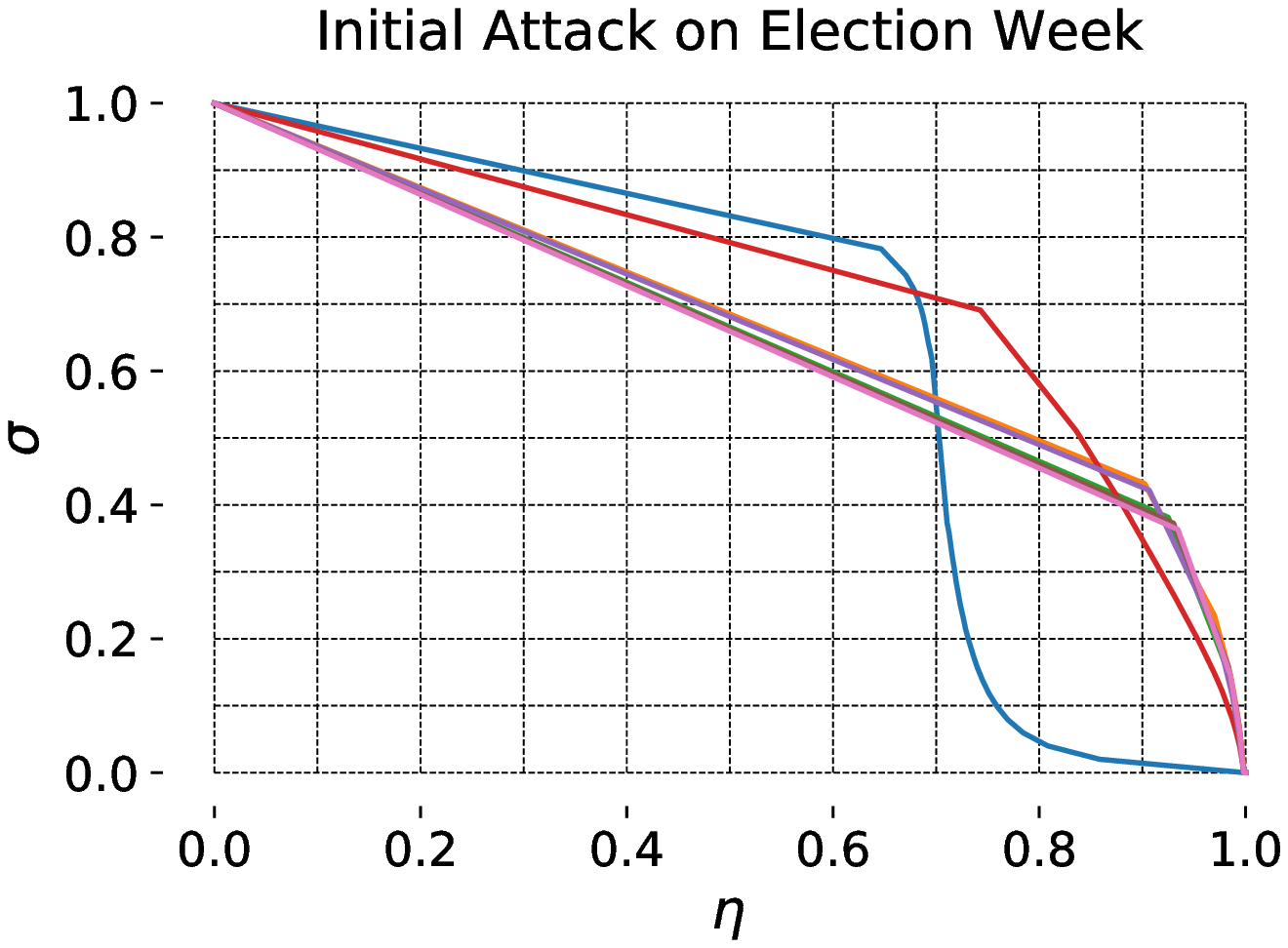} }}%
    \caption{Fragmentation of the Election-Week Network. The legend can be found in Figure \ref{fig:pa_roads}(b).}%
    \label{fig:election}%
\end{figure*}

At first this seems like a failure of the modularity vitality method. 
However, inspection of Figure \ref{fig:election} (b), shows otherwise.
In terms of links, the modularity vitality method is actually the most efficient attack strategy.
This counter-intuitive result occurs because \textit{none} of the methods are very effective at fragmenting the network.
The largest component is small when removing 10-30\% of nodes using the other methods, but those nodes account for over 95\% of the networks links.
The difference in bridge-first Modularity-Vitality attacks and all others, however, does highlight the fact that networks with extremely high-degree nodes will require mixed or hub-first approaches to be efficiently fragmented. 
Even accounting for this aspect of the network, the election week network exhibits extreme robustness to these types of attacks.

\begin{table*}[!htb]
    \caption{Results for initial attacks on the PA-Road Network and the Canadian-Election Twitter Network. Methods introduced in this work are on the left of the double column. The best results are emboldened.}
    \label{tab:cost_results}
    \centering
    \tiny
    \resizebox{\textwidth}{!}{%
    \begin{tabular}{|l|c|c||c|c|c|c|c|}
         \hline 
         \textbf{Network} & MV & AMV & AMC-D & Mas & CHB & WMC-D & Deg  \\
         \hline \hline
         PA-Roads $C_\rho$ & \textbf{0.013} & 0.016 & 0.015 & 0.167 & 0.162 & 0.120 & 0.122
         \\
         \hline
         PA-Roads $C_\eta$ & \textbf{0.021} & 0.026 & 0.026 & 0.295 & 0.281 & 0.262 & 0.305
         \\
         \hline
         \hline
         Election $C_\rho$ & 0.430 & 0.032 & \textbf{0.022} & 0.067 & 0.029 & 0.023 & \textbf{0.022}
         \\
         \hline
         Election $C_\eta$ & \textbf{0.635} & 0.673 & 0.656 & 0.732 & 0.667 & 0.654 & 0.651
         \\
         \hline
    \end{tabular}
    }
\end{table*}

\subsection{Discussion}
We see that modularity-based methods were very effective in all three studies.
The modularity vitality method shows that the PA Road network is over 8.5 times as fragile as could be seen with existing measures.
While the standard modularity vitality attack was effective on the PA-Road network, it was not on the Election week network.
However, using the absolute-value of modularity vitality resolves the issue.
This implies that attacking community-bridges is not enough.
By taking the absolute value, both community-bridges and community-hubs are attacked, leading to a method that is more robust across networks, even if it might not be the top-performer for specific networks.

As much as the values of a centrality are important, often the ranking of node centralities takes precedence.
This is the case with network attacks studied in this work.
So to go beyond the fragmentation results, the Kendall correlation of each method was calculated to compare the resulting node-rankings \cite{kermack1927contribution}.
Figures \ref{fig:pa_roads_corr} and \ref{fig:election_week_corr} show the correlations for the Road network and the Election network, respectively.
These correlations allow us to see the similarity of centrality ranking, regardless of the effectiveness of said ranking.
Though more clearly in Figure \ref{fig:pa_roads_corr}, we see that the existing degree-based metrics are highly correlated.
This is intuitive, as they are all alterations on a weighted degree.
While connecting certain groups might give a node a higher or lower score depending on the metric, a low degree usually leads to a low score.

\begin{figure}[!htb]
    \centering
    \includegraphics[width=\columnwidth]{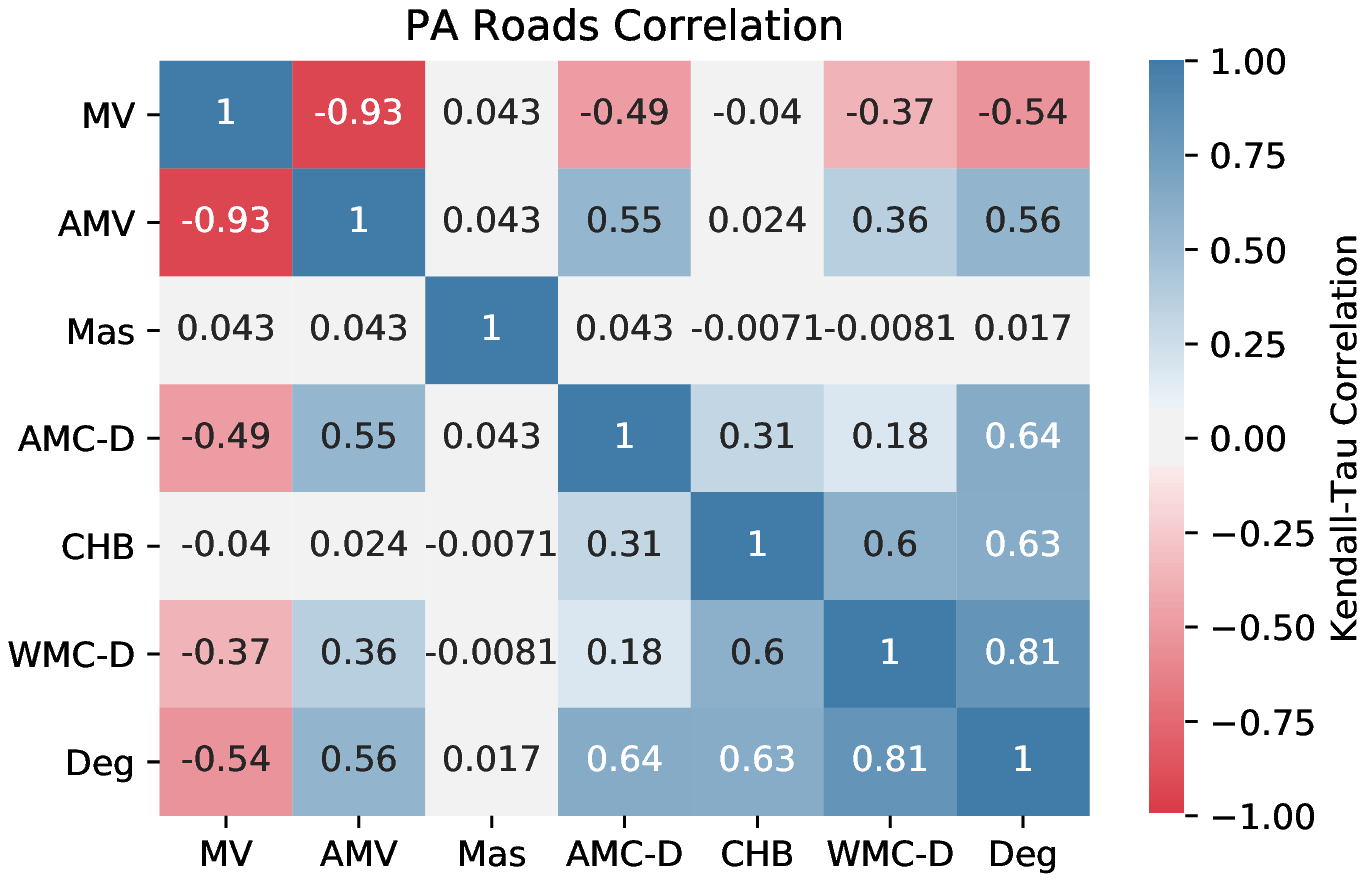}
    \caption{Kendall-Tau Correlation of the ``initial" attack strategies on the PA Roads Network.}
    \label{fig:pa_roads_corr}
\end{figure}

\begin{figure}[!htb]
    \centering
    \includegraphics[width=\columnwidth]{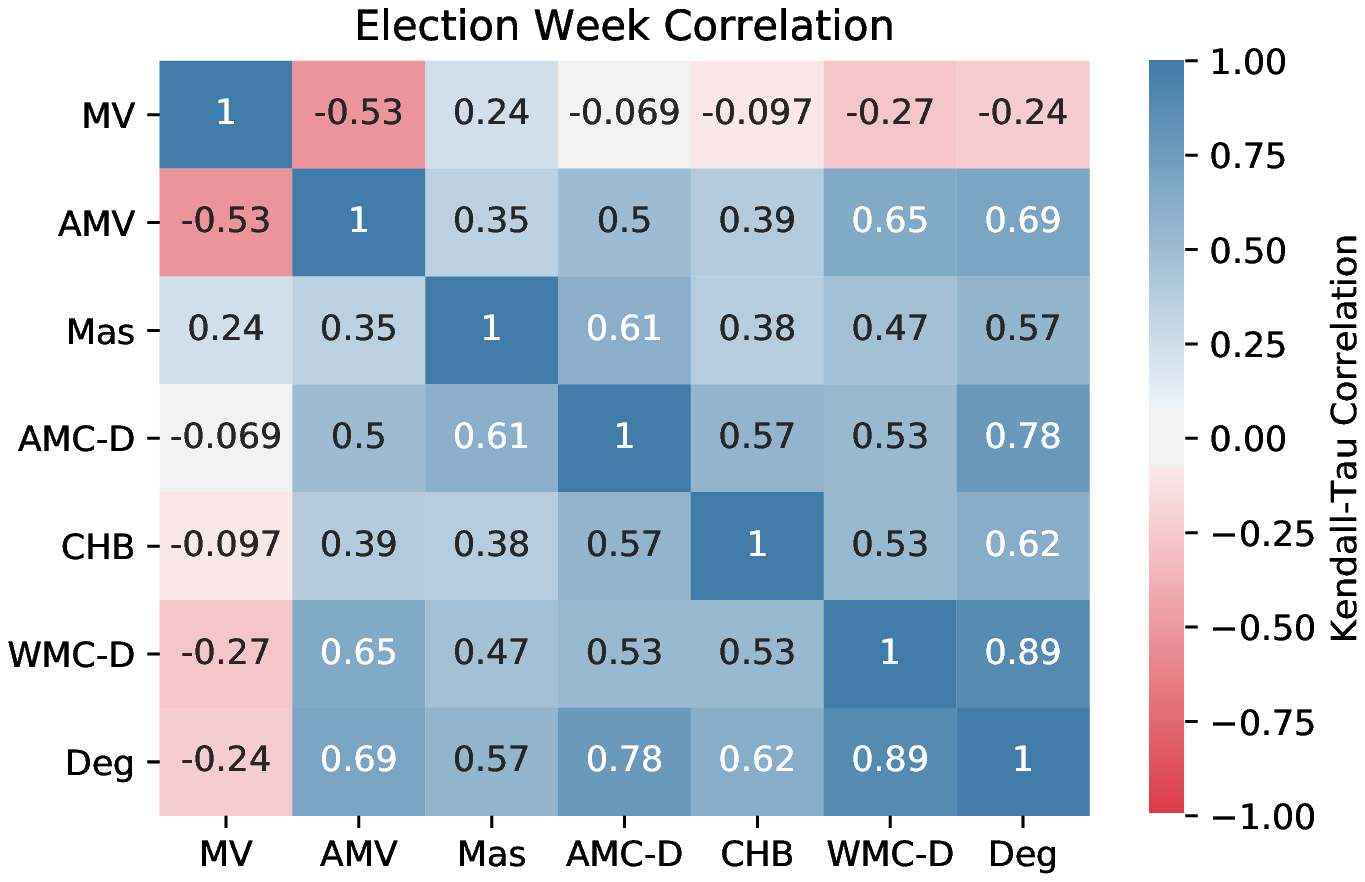}
    \caption{Kendall-Tau Correlation of the ``initial" attack strategies on the Election Week Network.}
    \label{fig:election_week_corr}
\end{figure}

Absolute modularity vitality has moderate correlation to the existing methods.
Most notably, it has strongest connections to the modular-degree centrality. 
However, the standard modularity vitality has lower correlation.
The combination of these observations show that modularity vitality is leveraging similar information to modular-centrality applied to degree, while giving those values a sign indicating the type of central role they are playing: hub or bridge.
The correlation between modularity vitality and its absolute value give further information about a network's structure.
In the road network, the strong negative correlation (-0.93) indicates that most nodes are community hubs, not bridges.
The same is seen with the election week network though to a lesser extent since the correlation is -0.53.
This result is consistent with the networks' high modularities, and that the road network's modularity is much higher than election week's.
This added information is a key contribution of the work, and will be of use for deep dives into network data.

Lastly, we see that our adjusted version of the modular-degree centrality gives improvements over the original modular-centrality, and that it has stronger correlations to the modularity-based methods.
Based on these results, it is possible that the generalized modular-centrality should also be adjusted to favor bridges.
In general, it seems that bridge-favoring methods have performed best in our experiments.
This is intuitive from a diffusion perspective.
If a network is highly modular, the groups themselves can act as silos to contain what is being diffused if the community-bridge nodes are removed.
For the road network, modular-centrality points to areas that need extra redundancy to create a more robust transportation network.

From the social network, we see that targeting bridges is not always enough. In the presence of community bridges and large community hubs, an approach that attacks both is necessary.
The absolute modularity vitality method attacks both, but the election network was robust to even this attack.

In the context of misinformation on social media, both users acting as hubs within fringe communities and users attempting to bridge communities play key roles.
Further, network robustness is both a strength and a weakness in this context.
A robust communication network means many users have the power to spread information.
This allows for distributed power of information but also means that user-based interventions to hamper the spread of misinformation will be ineffective.
It is commonly stated that network metrics may identify key points where misinformation diffusion can be stopped \cite{shin2018diffusion}.
However, we find that not to be the case.
The networks are too robust to have a number of points that control diffusion.
This may explain why disinformation tends to repeatedly resurface \cite{shin2018diffusion}.
While identifying key users spreading misinformation is useful for characterizing efforts to share fake news, we must look beyond user-based interventions to actually fight its spread.

\section{Community Deception}
The goal of community deception is to hide a community from detection algorithms \cite{fionda2017community,chen2019ga}.
The motivation behind this is typically to protect privacy.
Sensitive user data is often over-mined, and network community information is one of the ways in which identifiable information can be discovered.
The idea, then, is to alter the network such that community information is harmed, as measured through modularity of the original grouping on the altered network.

Previously,  modularity vitality attacks were used to maximize fragmentation.
However, fragmentation is only a by-product of the modularity vitality attack.
The attack's true objective is to maximize modularity.
As shown in Figure \ref{fig:election_mod_frag}, the same attack used to fragment the Election Week network increases its modularity.
In fact, all of the fragmentation methods increase modularity.
By attacking nodes which bridge communities, the communities become more separated and modularity increases.
The figure shows that the different attacks give similar change in modularity, though the vitality approach is most efficient, since it explicitly increases modularity.

\begin{figure*}[!ht]
    \centering
    \subfloat[Modularity by nodes removed.]{{\includegraphics[width=0.48\textwidth]{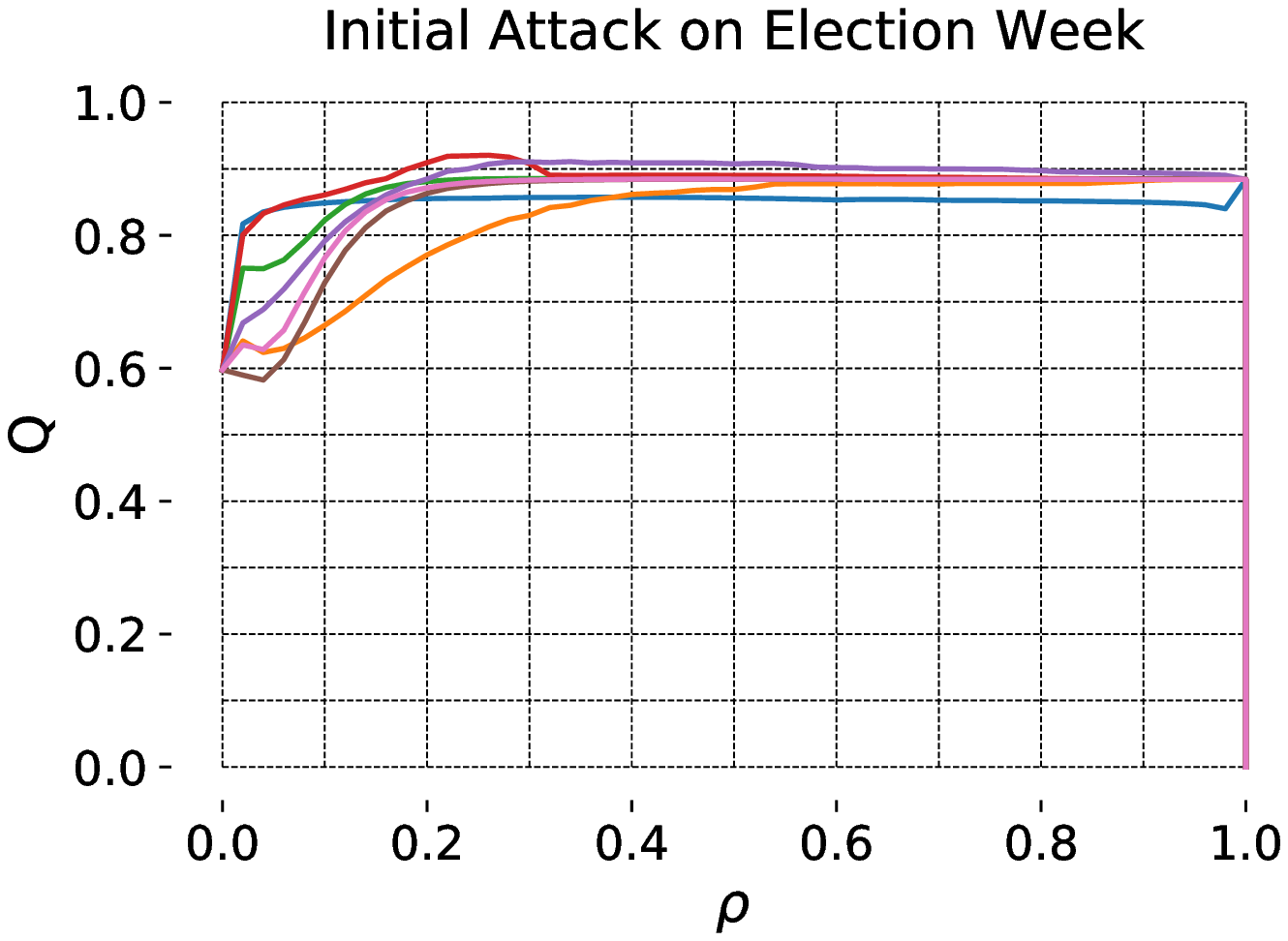} }}
    \subfloat[Modularity by edges removed.]{{\includegraphics[width=0.48\textwidth]{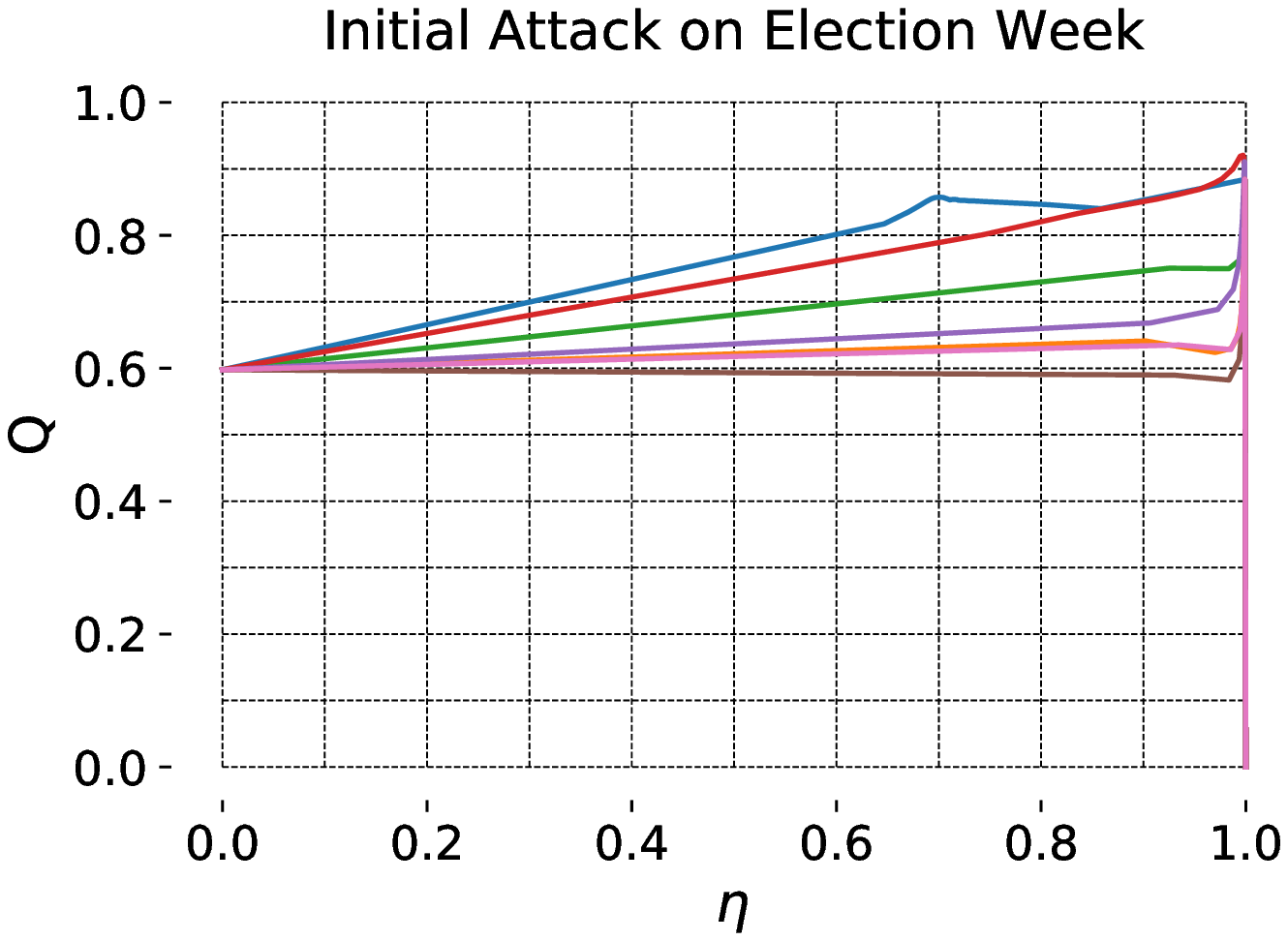} }}%
    \caption{Changes in modularity due to the fragmentation attacks, all using the initial strategy.}%
    \label{fig:election_mod_frag}%
\end{figure*}

For community deception, the power of the modularity vitality method is the ability to select community hubs instead of bridges.
Since community deception seeks to \textit{minimize} modularity, the attack can simply be reversed by selecting the node with the highest \textit{positive} modularity vitality.
Thus, a greedy solution to the node-based community deception problem is a recomputed, reversed, modularity vitality attack.
A faster approximation to this is the initial, reversed, modularity vitality attack.

Previous methods considered edge rewirings.
In practice, this may be difficult or problematic, since links in the altered network may or may not truly exist.
An alternate approach is to remove a small subset of the nodes.
While the altered network will have less links than the original, all links in the altered network are links in the original network.
By leveraging the modularity equation itself, we can select the nodes guaranteed to minimize modularity in a scalable way.
As a demonstration of this, community-deception was performed on the Canadian Election network, and the results are given in Figure \ref{fig:election_deception}, for both the fast approximation of the greedy approach.
For networks of this scale, even the greedy approach is very expensive.
Using the initial attack strategy, modularity can be dropped from approximately 0.7 to just over 0.4 by removing less than 2\% of nodes, as shown in Figure \ref{fig:election_deception} (a).
However, Figure \ref{fig:election_deception} (b) shows that this comes at a cost of 45\% of the nodes edges.
Modularity can be decreased further, though with diminishing returns.
Modularity levels out when about 8\% of nodes and 50\% of edges are removed, resulting in a final modularity of 0.36, which is a 49\% decrease.

We know from the modularity vitality equation that this strategy is attacking hubs, and this is seen by the fact that the first 2\% of nodes targeted are accounting for 45\% of links.
Intuitively, this suggests that a user's connections to Twitter accounts that are popular within a community reveal that user's identity as a community member.
If this identity is to be protected, then hiding these key hubs, as measured through modularity vitality, is the most effective strategy.

This presents a dilemma to social media users wishing to conceal their online community: the most effective strategy is to un-friend or un-follow the community's leaders, which would undoubtedly harm the community itself.
The extent of this harm is dependent on the platform. 
On Twitter, for example, users may interact without a following relationship.
On other platforms, like Facebook, the extent of these interactions is more limited.
This leaves it up to the social media companies to protect their users by allowing them to hide their affiliations to other accounts, or at least to community leaders.

\begin{figure*}[!ht]
    \centering
    \subfloat[Modularity by nodes removed.]{{\includegraphics[width=0.48\textwidth]{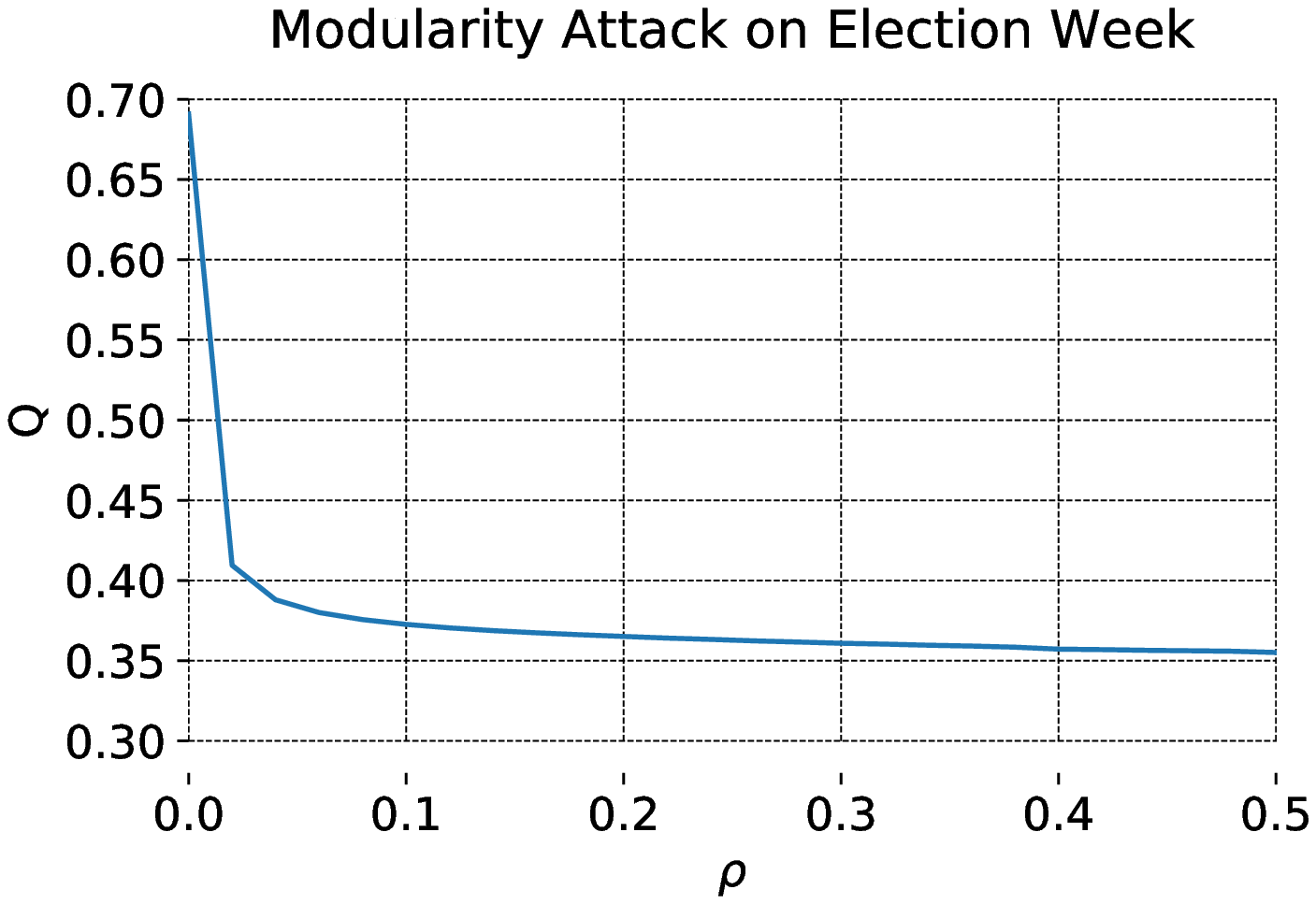} }}
    \subfloat[Modularity by edges removed.]{{\includegraphics[width=0.48\textwidth]{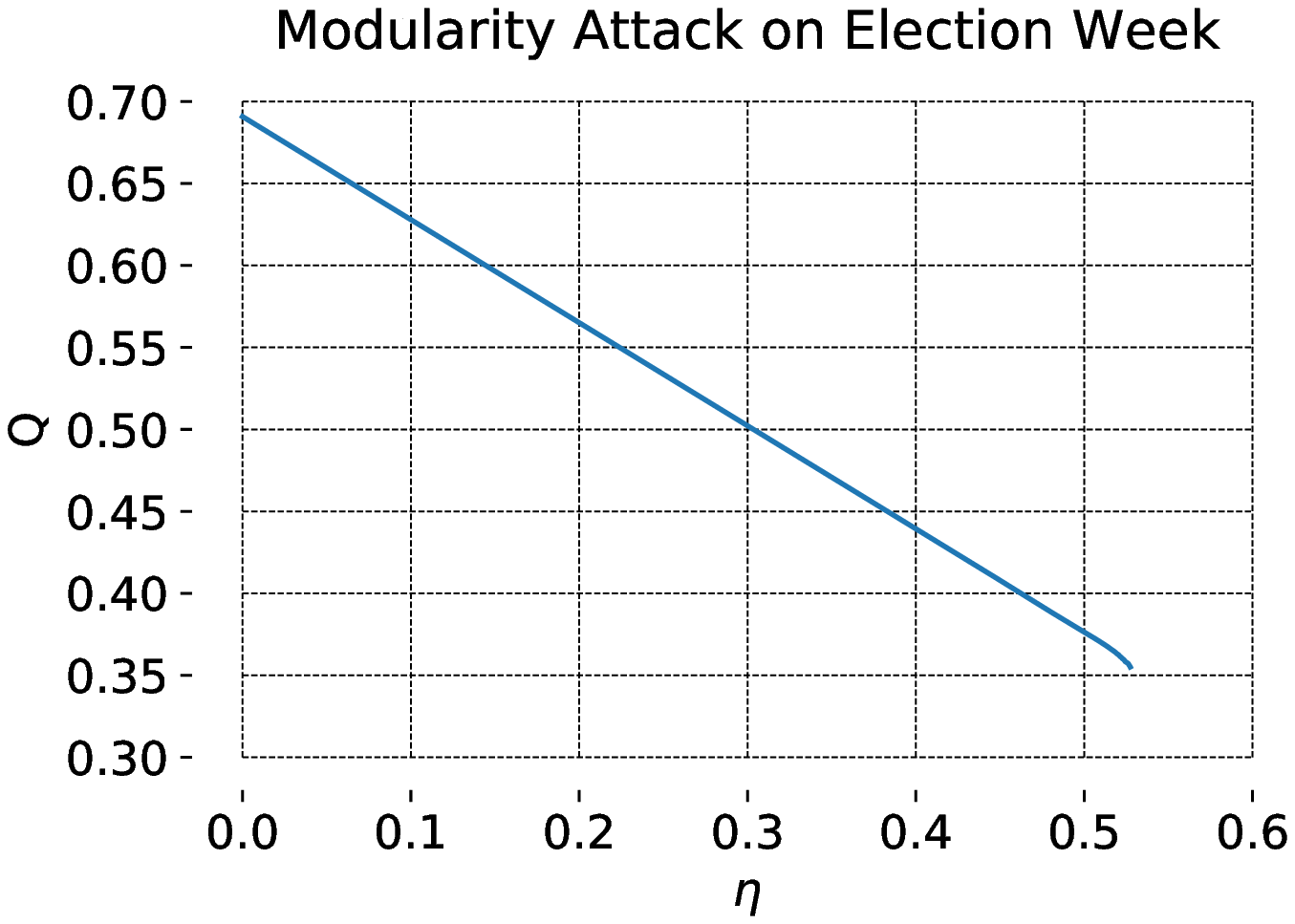} }}%
    \caption{Community deception on the Election-Week Network using the initial attack strategy.}%
    \label{fig:election_deception}%
\end{figure*}

A choice must be made when performing community-deception: At what point is does the cost of deleting network data outweigh the benefit of decreased modularity?
For this case, if only nodes are of interest, there is only a very small price to pay to decrease modularity by 41\%.
If edges are important to consider, the cost is higher.
This is only an approximation of the greedy approach.
The greedy approach, recomputed reversed modularity vitality attack, will likely achieve even better results.
A study of this comparison on smaller networks along with non-greedy alternatives is left for future work.
Additionally, this type of attack could be combined with the previously studied edge re-wiring attacks to give even obscure communities even more effectively.
Lastly, explicit modularity \textit{maximization} through node removal could have interesting applications, such as node filtering to obtain more interpretable groups. 
This, too, is left for future work.

\section{Conclusion}
Both centrality measures and community detection are core research areas in Network Science.
At the intersection of these areas, community-aware centrality measures attempt to quantify how central nodes are given a network partition.
Though the areas are closely related, the current community-aware centralities are not strongly tied to community theory.
Here, we examine modularity vitality, which measures the change in modularity of a network and its partition if a node were to be removed.
Thus, modularity vitality measures each node's individual contribution to group structure.
This measure is directly derived from the modularity equation, giving the measure a strong link to community detection theory.
We derive a scalable method of calculating modularity vitality, which improves over the naive method usually by a factor of $N$, allowing for the analysis of massive networks.

Unlike existing measures, however, modularity vitality not only quantifies how important a node is, but in which way it is important.
Once groups are introduced, nodes can take on two central roles: hubs within their community, and bridges between communities. 
The role is encoded in the sign of modularity vitality; nodes with negative values are bridges, while positive-valued nodes are hubs.

Modularity vitality was tested in three settings: generated cellular networks, the Pennsylvania Road Network, and a Twitter network capturing conversation around the Canadian Election of 2019.
In these tests, we saw that modularity-based methods outperformed existing community-aware centralities as measured through network fragmentation.
Our results show that the Pennsylvania Road network is over 8.5 times more fragile than the existing measures would have concluded, and that community bridges play a more important role than community-hubs.

Further, we saw that the social media conversation network is very robust, and that both community-hubs and community-bridges play important roles in that robustness.
Additionally, the presence of extremely high-degree nodes lead to bridge-first methods performing worst, since high-degree nodes are typically well-grouped.
Robust communication is aligned with Social Media's business interests, since they give many users the potential to ``go-viral," encouraging engagement.
The specific source of this robustness remains an area of future research, though the balance of nodes with positive and negative modularity vitality nodes suggests that the presence of many community bridges may be a factor.
This theory is in agreement with the results on the PA network, which has very few bridges and is extremely fragile.
A robust communication network suggests that user-based interventions are not an effective strategy to fight the spread of misinformation, since an extreme intervention like user-removal only has a small impact on potential diffusion.

Many prior community-aware centralities give preference to community-bridges over community-hubs.
Using modularity-vitality without taking the absolute value also targets bridges instead of hubs.
Based on this, we include a modified version of Ghalmane's generalized community-aware centrality measure where bridges are favored instead of hubs.
This alternate version of their community-aware centrality when applied with degree performed better in our experiments.
Further studies could explore if this change is an improvement when combined with classical centrality measures other than degree.

Lastly, we recognize that modularity vitality can be used as a greedy solution to the community-deception problem.
Community-deception seeks to remove nodes or edges to maximally reduce modularity, which could be important for privacy protection in data distribution.
While previous work uses a genetic algorithm to select nodes or edges which may reduce modularity, modularity vitality can be used to select the node that will maximally decrease modularity.
Recomputing modularity vitality at each removal provides a greedy solution to the community-deception problem, but we use the faster approximation: only calculating modularity vitality once.
While the genetic algorithm could scale to networks with two hundred nodes, the approximation of the greedy method scales to networks with millions of nodes and hundreds of millions of links, as demonstrated on the election week network.
Through this demonstration we see that modularity can be decreased by 41\% while only removing less than 2\% of nodes, but this comes at a cost of 45\% of the edges.
Still, community-deception is a combinatorial optimization problem, so there are almost definitely better solutions.
Going forward, the greedy approach using modularity vitality may be a useful baseline.

The findings suggest that the most effective strategy currently available to users attempting to protect their community identity is to remove their connections to community leaders. 
This strategy clearly will negatively impact the community itself, leaving users with little options to protect their privacy. 
It is up to social media companies to protect this privacy by allowing users to hide their connections.

We have demonstrated that modularity vitality is a powerful method of finding nodes that bridge communities or are hubs within their communities at scale.
Modularity is but one of many cluster evaluation functions.
Exploration of vitalities of these other functions could give an alternative view of nodal contributions to community structure.
Community quality vitalities, and community-aware centralities more generally have many applications to areas such as infrastructure robustness, traffic improvement, immunization, and social media.
Deeper dives into these application areas using the techniques proposed here could be fruitful areas of future research.

\ifCLASSOPTIONcompsoc
  \section*{Acknowledgments}
\else
  \section*{Acknowledgment}
\fi
This work was supported in part by the Office of Naval Research (ONR) Award N000141512797 Minerva award for Dynamic Statistical Network Informatics, ONR Award N000141712675, and the Center for Computational Analysis of Social and Organization Systems (CASOS). Thomas Magelinski was also supported by an ARCS Foundation scholarship. The authors would also like to thank the anonymous reviewers for their thoughtful and detailed feedback. The views and conclusions contained in this document are those of the authors and should not be interpreted as representing the official policies, either expressed or implied, of the ONR.

\ifCLASSOPTIONcaptionsoff
  \newpage
\fi



%


\bibliographystyle{IEEEtran}
\bibliography{references.bib}

%

\begin{IEEEbiography}[{\includegraphics[width=1in,height=1.25in,clip,keepaspectratio]{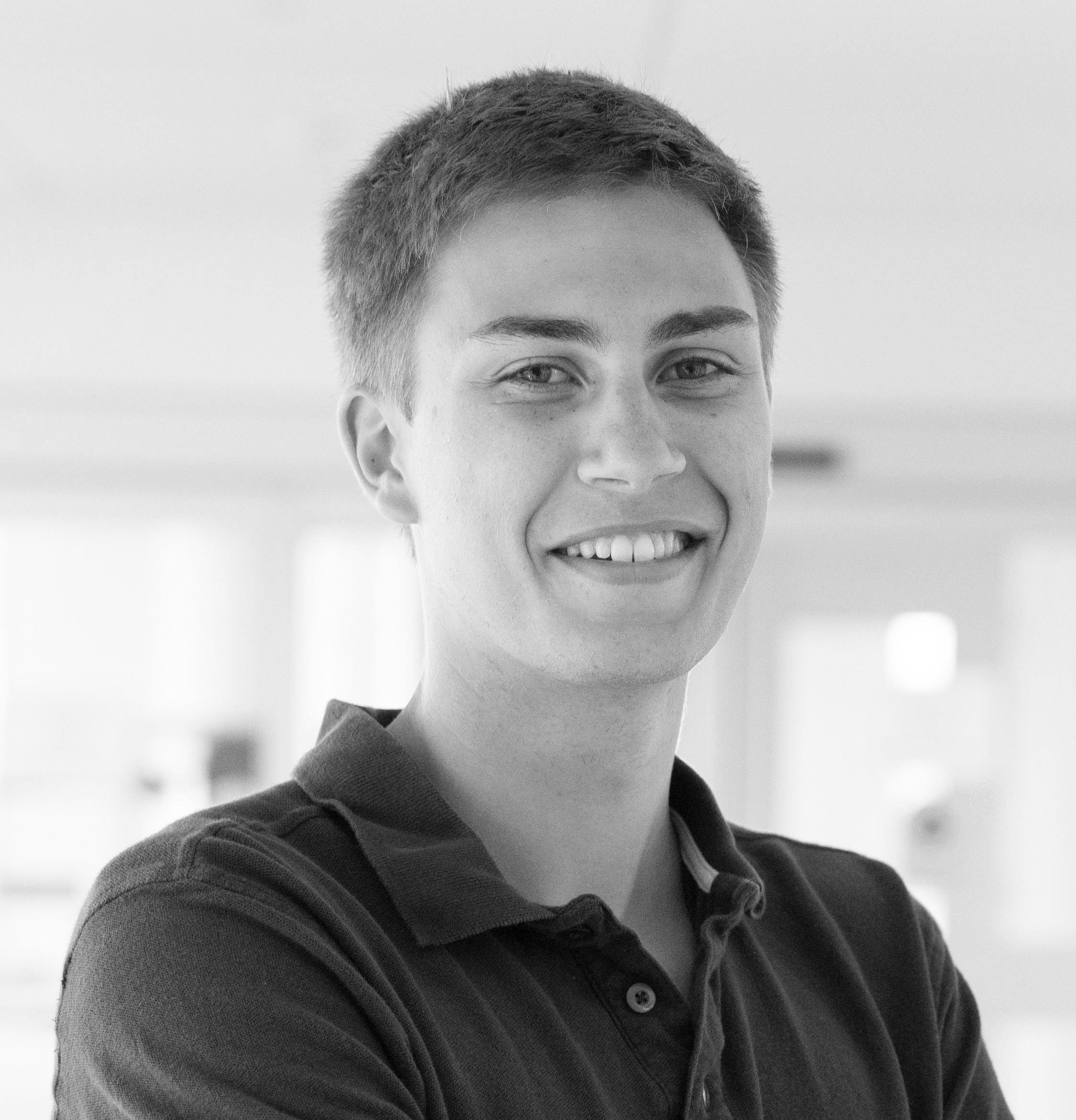}}]{Thomas Magelinski}
Thomas Magelinski received the B.S. degree in engineering science and mechanics from Virginia Tech, Blacksburg, VA, USA, in 2017. He is currently pursuing the Ph.D. degree in societal computing with Carnegie Mellon University, Pittsburgh, PA, USA, under the supervision of Prof. K. M. Carley.
His current research interests include network science and the dynamics of complex systems.  
\end{IEEEbiography}

\begin{IEEEbiography}[{\includegraphics[width=1in,height=1.25in,clip,keepaspectratio]{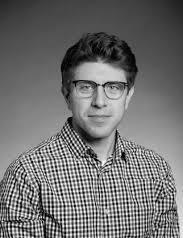}}]{Mihovil Bartulovic}
Mihovil Bartulovic received the M.S. degree from Carnegie Mellon University, Pittsburgh, PA, USA in 2017 in Electrical and Computer Engineering, and a M.S. degree from the University of Zagreb, Zagreb, Croatia in 2014 in Control Engineering and Automation, and a B.S. degree from University of Zagreb, Zagreb, Croatia in 2012 in Electrical Engineering and Information Technology. He is currently a Ph.D. student at Carnegie Mellon University in the Electrical and Computer Engineering department in the College of Engineering, in Pittsburgh, PA USA, under the supervision of Prof. K. M. Carley.
His current research interests include network science, sequential network data and trails.
\end{IEEEbiography}


\begin{IEEEbiography}[{\includegraphics[width=1in,height=1.25in,clip,keepaspectratio]{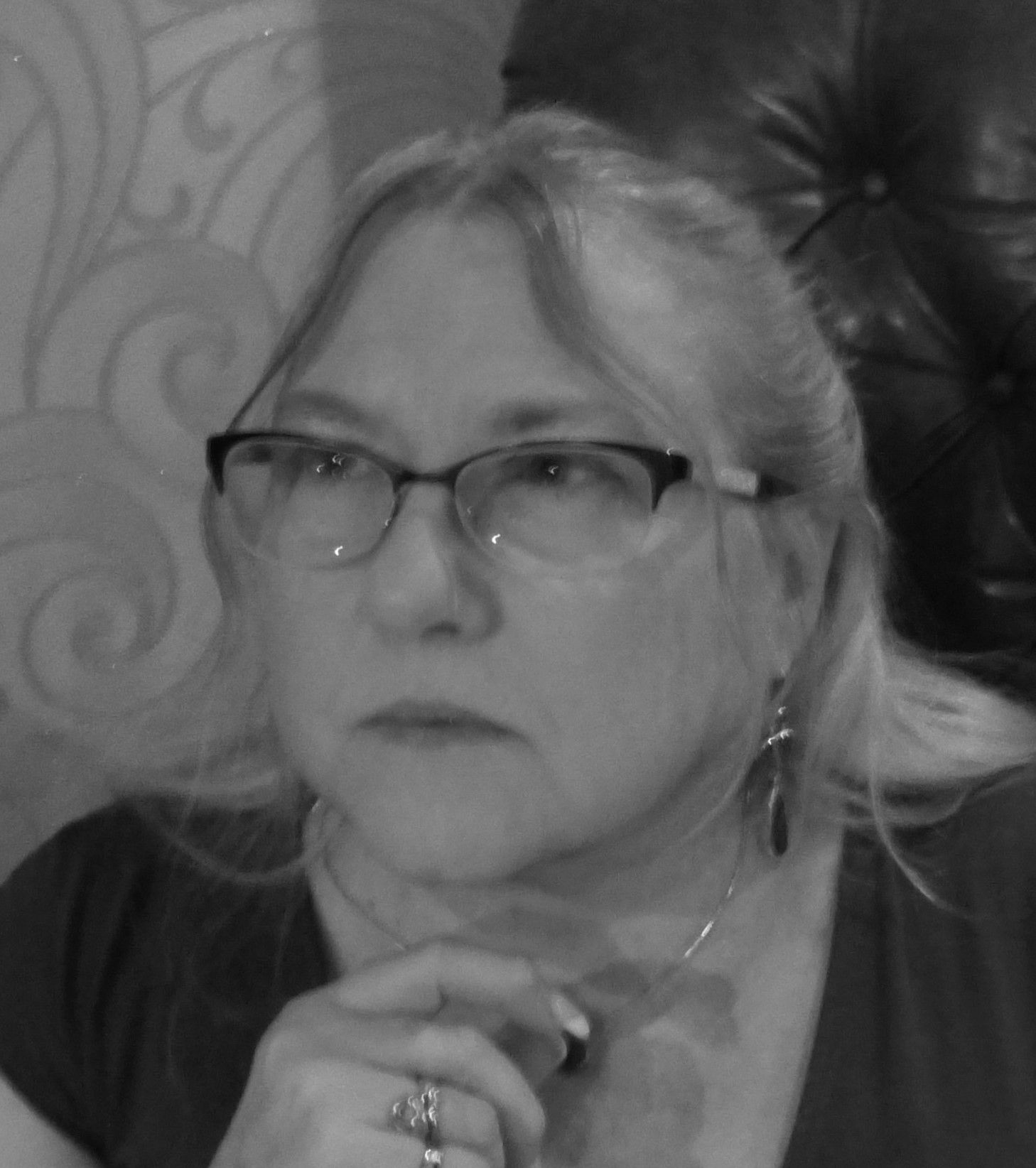}}]{Kathleen M. Carley}
Kathleen M. Carley received the Ph.D. degree from Harvard University, Cambridge, MA, USA in 1984 in sociology, and two S.B. degrees from the Massachusetts Institute of Technology, Cambridge, MA, USA in 1978 in political science and in economics, and an H.D. from the University of Zurich.
She is a professor at Carnegie Mellon University in the Institute of Software Research department in the School of Computer Science, in Pittsburgh, PA USA.
She directs the center for Computational Analysis of Social and Organizational Systems (CASOS) and the center for Informed Democracy and Social-cybersecurity (IDeaS) at Carnegie Mellon University.
She is the CEO and President of Carley Technologies Inc, A.D.B. as Netanomics, in Sewickley, PA, USA. 
Her research examines complex socio-technical systems using high-dimensional dynamic networks, agent based models, and text mining.
Prof. Carley is a member of the Association for Computing Machinery, the Academy of Management, INFORMS, the American Association for the Advancement of Science, the International Network for Social Network Analysis and the American Sociological Association.
She is the editor in chief of Computational, Mathematical and Organizational Theory, and serves on numerous editorial boards.
She is an IEEE Fellow, and an associate editor for IEEE Transactions on Computational Social Systems.
\end{IEEEbiography}



\clearpage
\input{appendix.tex}

\end{document}

%% file: appendix.tex
\appendices

\section{Community-Degree}
\label{sec:community_degree}
Though the signed aspect of modularity vitality is quite useful, it is possible that a node has high positive and negative components of modularity in Equation \ref{eqn:new_modularity}, resulting in a modularity vitality near zero.
These nodes may be particularly important for networks with low modularity.
We can adjust Equation \ref{eqn:new_modularity} to obtain a measure which credits nodes for hub \textit{and} bridge behavior. By changing the subtraction of $h_{i,c}$ to addition, this effect is achieved. After this adjustment, there is no need for a separate internal term, making the final measure:
\begin{equation}
    CD_i = \frac1{4 \left ( M-k_{i}\right )^2} \sum_{c \in \mathfrak{C}} \left ( d_c + h_{i,c}\right )^2
\end{equation}
Again, attachment to large groups is favored over attachment to small groups.
Since this is just weighting the degree, we will call it Community-Degree (CD).
The previous results including this measure are shown in Tables \ref{tab:generated_results_extended}-\ref{tab:generated_results_er}, and in Figures \ref{fig:pa_roads_corr_extended} and \ref{fig:election_week_corr_extended}.

\begin{figure}[!htb]
    \centering
    \includegraphics[width=\columnwidth]{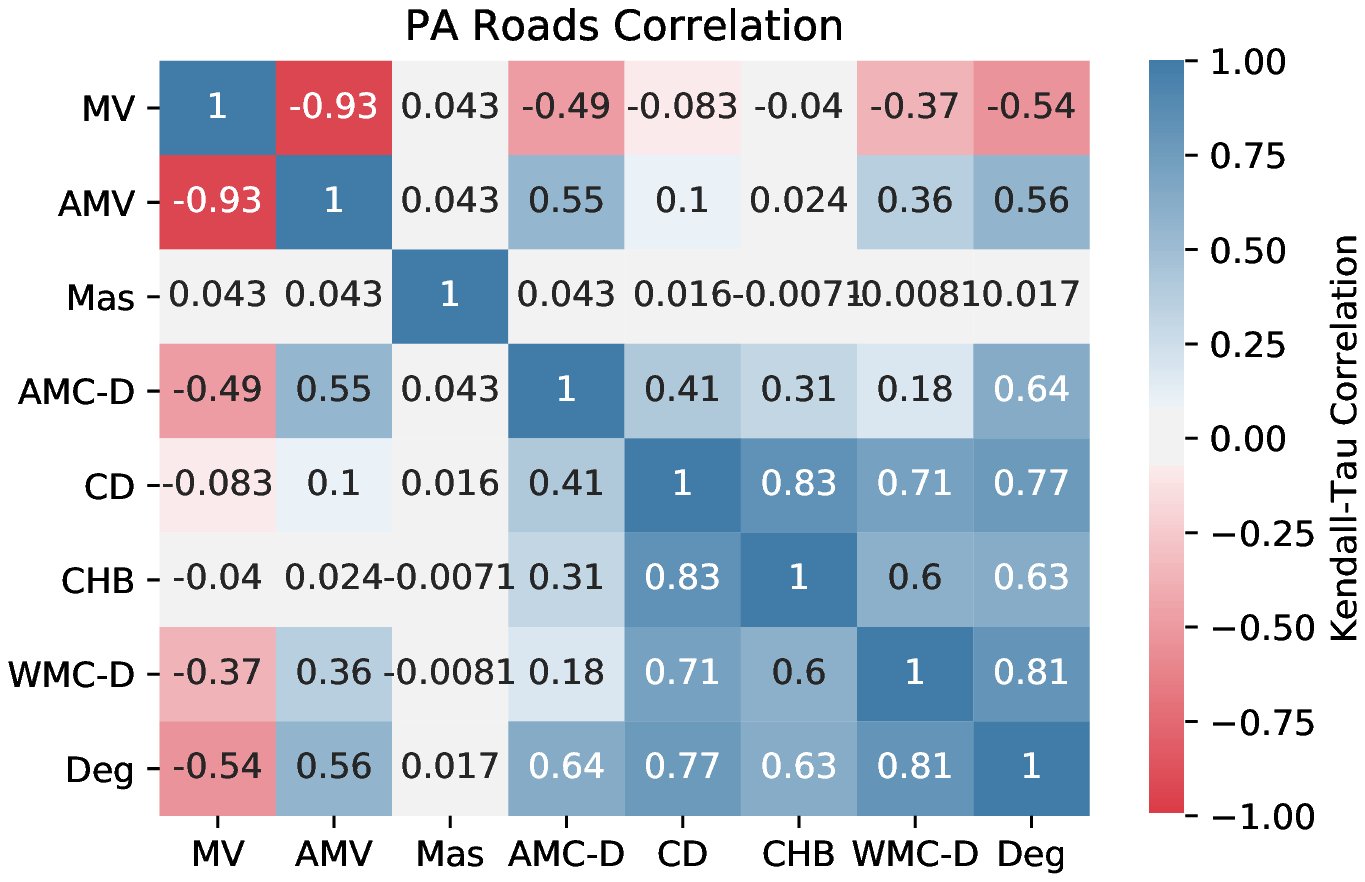}
    \caption{Extended version of Figure \ref{fig:pa_roads_corr} to include Community-Degree. Kendall-Tau Correlation of the ``initial" attack strategies on the PA Roads Network.}
    \label{fig:pa_roads_corr_extended}
\end{figure}

\begin{figure}[!htb]
    \centering
    \includegraphics[width=\columnwidth]{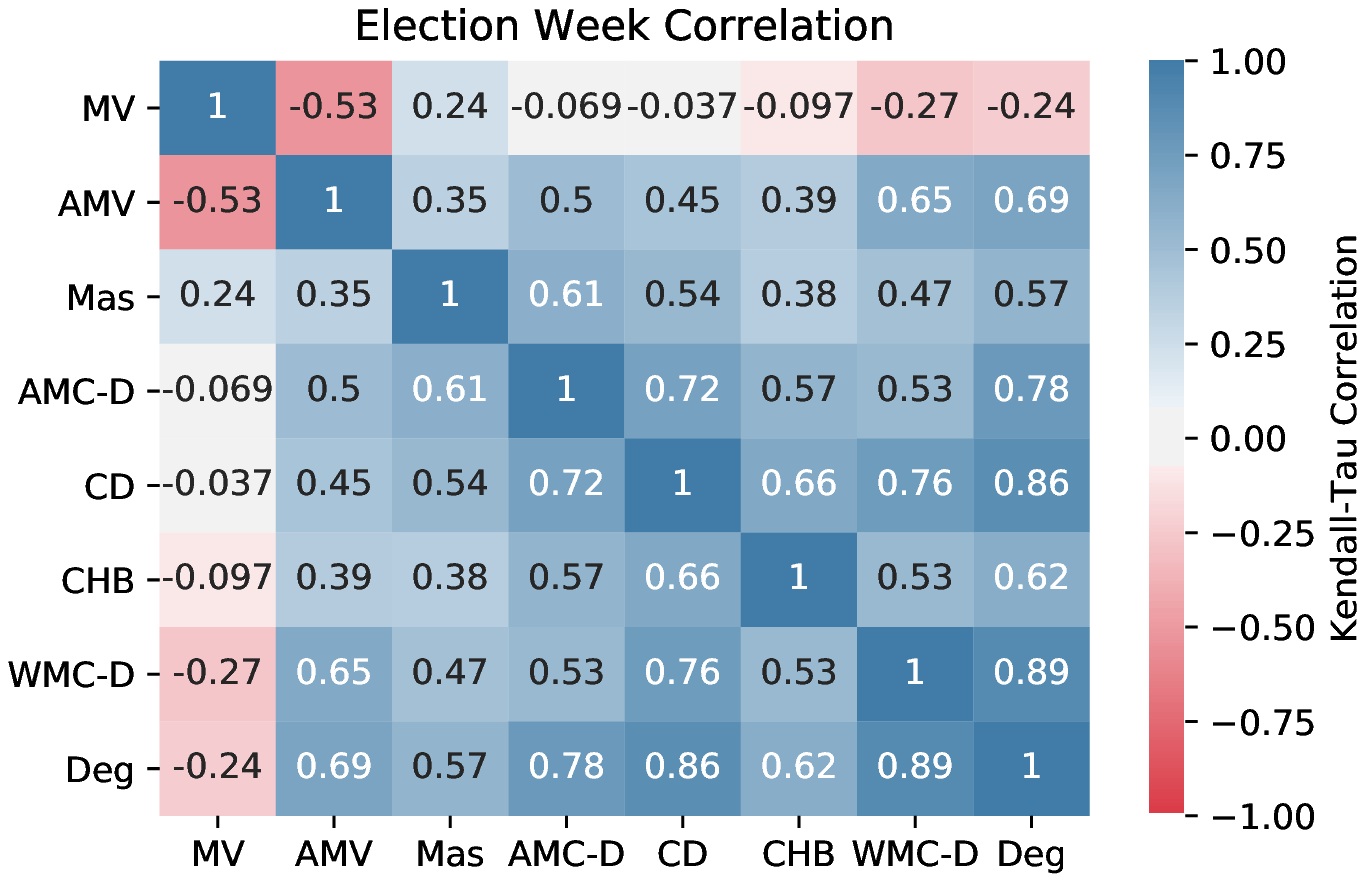}
    \caption{Extended version of Figure \ref{fig:election_week_corr} to include Community-Degree. Kendall-Tau Correlation of the ``initial" attack strategies on the PA Roads Network.}
    \label{fig:election_week_corr_extended}
\end{figure}

Community-Degree is highly correlated with degree, and so performs similarly.
Based on these results, it seems that the signed centrality is more effective while also conveying more information.

\section{Results on Other Generated Networks}
For completeness, networks lacking strong group structure were generated.
Scale-free networks were generated using the Barab\'{a}si-Albert model using parameters $n=1000,$ $m=8,$ and $\gamma=1.5.$
Over the 100 iterations tested the average modularity from Leiden grouping was 0.196.
The results are given in Table \ref{tab:generated_results_sf}.

Erd\H{o}s-R\'enyi networks with parameters $n=1000,$ $p=0.015$, were also created.
These parameters were chosen to give similar density to the cellular networks previously studied.
Networks were generated until a connected network was reached.
Over the 100 connected networks, the average modularity from Leiden grouping was 0.240.
The results are given in Table \ref{tab:generated_results_er}.

The results across network types are similar.
First, none of the attacks are very effective.
Both the node and edge cost are higher than that seen for the Election network, which was robust.
With that said, the degree and modular-degree attacks were consistently the most efficient in terms of nodes.
This is intuitive; without more meaningful structure, the most effective strategy is to look at the node with the most edges.
This results in a high edge-based cost, however.
So we see that modularity vitality actually performs best in terms of edge-cost.
Lastly, we see that the adjusted-modular degree that we proposed performs similarly to the original.
Adjusted measure performs much better on highly modular networks, while performing similarly on less modular networks.

\begin{table*}[!htb]
    \caption{Extended version of Table \ref{tab:generated_results} to include Community-Degree. Results for attacks on the generated cellular networks. The values shown are the average over 100 simulations. Methods introduced in this work are on the left of the double column. The best results are emboldened.}
    \label{tab:generated_results_extended}
    \centering
    \tiny
    \resizebox{\textwidth}{!}{%
    \begin{tabular}{|l|c|c|c||c|c|c|c|c|}
         \hline
         \textbf{Method} & MV & AMV & CD & AMC-D & Mas & CHB & WMC-D & Deg  \\
         \hline \hline
         Initial $C_\rho$ & \textbf{0.165} & 0.211 & 0.361 & 0.169 & 0.198 & 0.383 & 0.381 & 0.347
         \\
         \hline
         Initial $C_\eta$ & \textbf{0.247} & 0.308 & 0.560 & 0.268 & 0.293 & 0.576 & 0.599 & 0.578
         \\
         \hline
         \hline
         MBA $C_\rho$ & \textbf{0.086} & 0.087 & 0.099 & 0.088 & 0.090 & 0.101 & 0.103 & 0.100
         \\
         \hline
         MBA $C_\eta$ & \textbf{0.157} & 0.162 & 0.210 & 0.173 & 0.162 & 0.211 & 0.219 & 0.216
         \\
         \hline
         \hline
         Recomputed $C_\rho$ & \textbf{0.107} & 0.126 & 0.320 & 0.130 & 0.132 & 0.331 & 0.337 & 0.309
         \\
         \hline
         Recomputed $C_\eta$ & \textbf{0.188} & 0.205 & 0.599 & 0.221 & 0.205 & 0.608 & 0.616 & 0.586
         \\
         \hline
    \end{tabular}
    }
\end{table*}

\begin{table*}[!htb]
    \caption{Extended version of Table \ref{tab:cost_results} to include Community-Degree. Results for initial attacks on the PA-Road Network and the Canadian-Election Twitter Network. Methods introduced in this work are on the left of the double column. The best results are emboldened.}
    \label{tab:cost_results_extended}
    \centering
    \tiny
    \resizebox{\textwidth}{!}{%
    \begin{tabular}{|l|c|c|c||c|c|c|c|c|}
         \hline
         \textbf{Network} & MV & AMV & CD & AMC-D & Mas & CHB & WMC-D & Deg  \\
         \hline \hline
         PA-Roads $C_\rho$ & \textbf{0.013} & 0.016 & 0.126 & 0.015 & 0.167 & 0.162 & 0.120 & 0.122
         \\
         \hline
         PA-Roads $C_\eta$ & \textbf{0.021} & 0.026 & 0.264 & 0.026 & 0.295 & 0.281 & 0.262 & 0.305
         \\
         \hline
         \hline
         Election $C_\rho$ & 0.430 & 0.032 & 0.023 & \textbf{0.022} & 0.067 & 0.029 & 0.023 & \textbf{0.022}
         \\
         \hline
         Election $C_\eta$ & \textbf{0.636} & 0.673 & 0.661 & 0.656 & 0.732 & 0.667 & 0.654 & 0.651
         \\
         \hline
    \end{tabular}
    }
\end{table*}

\begin{table*}[!hbt]
    \caption{Results for attacks on the generated scale free networks. The values shown are the average over 100 simulations. Methods introduced in this work are on the left of the double column. The best results by method are emboldened. The best results overall are marked with a star.}
    \label{tab:generated_results_sf}
    \centering
    \resizebox{\textwidth}{!}{%
    \begin{tabular}{|l|c|c|c|c||c|c|c|c|}
         \hline
         \textbf{Method} & MV & AMV & CD & AMC-D & Mas & CHB & MC-D & Deg  \\
         \hline \hline
         Initial $C_\rho$ & 0.483 & 0.424 & 0.263 & 0.256 & 0.337 & 0.361 & 0.254 & \textbf{0.243}
         \\
         \hline
         Initial $C_\eta$ & $\textbf{0.834}^*$ & 0.856 & 0.882 & 0.881 & 0.884 & 0.879 & 0.881 & 0.880
         \\
         \hline
         \hline
         MBA $C_\rho$ & 0.430 & 0.364 & 0.243 & 0.239 & 0.273 & 0.292 & 0.242 & \textbf{0.235}
         \\
         \hline
         MBA $C_\eta$ & \textbf{0.839} & 0.859 & 0.880 & 0.880 & 0.881 & 0.877 & 0.880 & 0.880
         \\
         \hline
         \hline
         Recomputed $C_\rho$ & 0.296 & 0.305 & 0.224 & 0.227 & 0.256 & 0.258 & $\textbf{0.223}^*$ & $\textbf{0.223}^*$
         \\
         \hline
         Recomputed $C_\eta$ & \textbf{0.878} & \textbf{0.878} & 0.880 & 0.880 & 0.882 & 0.881 & 0.880 & 0.880
         \\
         \hline
    \end{tabular}
    }
\end{table*}

\begin{table*}[!hbt]
    \caption{Results for attacks on the generated Erd\H{o}s-R\'enyi networks. The values shown are the average over 100 simulations. Methods introduced in this work are on the left of the double column. The best results by method are emboldened. The best results overall are marked with a star.}
    \label{tab:generated_results_er}
    \centering
    \resizebox{\textwidth}{!}{%
    \begin{tabular}{|l|c|c|c||c|c|c|c|c|}
         \hline
         \textbf{Method} & MV & AMV & CD & AMC-D & Mas & CHB & WMC-D & Deg  \\
         \hline \hline
         Initial $C_\rho$ & 0.493 & 0.491 & 0.479 & 0.475 & 0.487 & 0.492 & 0.473 & \textbf{0.472}
         \\
         \hline
         Initial $C_\eta$ & 0.683 & 0.681 & 0.715 & 0.723 & 0.705 & $\textbf{0.675}^*$ & 0.724 & 0.728
         \\
         \hline
         \hline
         MBA $C_\rho$ & 0.483 & 0.484 & 0.469 & 0.466 & 0.475 & 0.485 & 0.464 & \textbf{0.462}
         \\
         \hline
         MBA $C_\eta$ & 0.683 & 0.681 & 0.714 & 0.722 & 0.705 & $\textbf{0.675}^*$ & 0.724 & 0.727
         \\
         \hline
         \hline
         Recomputed $C_\rho$ & 0.461 & 0.482 & $\textbf{0.429}^*$ & 0.454 & 0.452 & 0.446 & 0.430 & 0.430
         \\
         \hline
         Recomputed $C_\eta$ & 0.700 & \textbf{0.681} & 0.739 & 0.739 & 0.729 & 0.718 & 0.738 & 0.740
         \\
         \hline
    \end{tabular}
    }
\end{table*}